\documentclass[12pt]{article}

\usepackage[utf8]{inputenc}
\usepackage[english]{babel}

\usepackage[showframe=false, a4paper, hmargin = 2.9cm, vmargin = 3cm]{geometry}
\usepackage[svgnames]{xcolor}

\usepackage{enumitem}
\usepackage{authblk}

\usepackage{graphicx}
\usepackage{subcaption}
\usepackage[width = 0.95\textwidth ,font = {small}]{caption}
\usepackage[above,below]{placeins}
\usepackage{float}

\usepackage[square,comma,numbers,sort&compress,merge]{natbib}
\usepackage{IEEEtrantools}
\usepackage{amsthm}
\usepackage{amsmath}
\usepackage{amsfonts}
\usepackage{amssymb}
\usepackage{cases}

\newtheorem{definition}{Definition}
\newtheorem{prop}{Proposition}
\newtheorem{theorem}[prop]{Theorem}
\newtheorem{lemma}[prop]{Lemma}
\newtheorem{corr}[prop]{Corollary}

\newcounter{rem}

\newfloat{bbox}{thp}{lop}
\floatname{bbox}{Box}
\usepackage{algorithm}
\floatname{algorithm}{Protocol}
\usepackage{algcompatible}

\usepackage[nottoc]{tocbibind}
\usepackage[colorlinks=true,urlcolor=DarkGreen,linkcolor=DarkGreen,citecolor=ForestGreen]{hyperref}
\input{math_notation}
\def\ieeetit{IEEE Trans.\ Inf.\ Th.}%

\def\jsovmath{J.\ Sov.\ Math.}
\def\nature{Nature}%
\def\njp{New J.\ Phys.}%
\def\npjqi{Npj Quantum Information}
\def\pra{Phys.\ Rev.\ A}%
\def\prl{Phys.\ Rev.\ Lett.}%
\def\quantum{Quantum}%
\def\rmp{Rev.\ Mod.\ Phys.}%

\begin{document}

\author{Thomas Van Himbeeck}
\author{Stefano Pironio}
\affil{Laboratoire d'Information Quantique, Universit\'e libre de Bruxelles (ULB), Belgium}

\title{Correlations and randomness generation\\
 based on energy constraints}
\date{May 22, 2019}

\maketitle

\begin{abstract}
In a previous paper, we introduced a semi-device-independent scheme consisting of an untrusted source sending quantum states to an untrusted measuring device, with the sole assumption that the average energy of the states emitted by the source is bounded. Given this energy constraint, we showed that certain correlations between the source and the measuring device can only occur if the outcomes of the measurement are non-deterministic, i.e., these correlations certify the presence of randomness.

In the present paper, we go further and show how to quantify the randomness as a function of the correlations and prove the soundness of a QRNG protocol exploiting this relation. For this purpose, we introduce (1) a semidefinite characterization of the set of quantum correlations, (2) an algorithm to lower-bound the Shannon entropy as a function of the correlations and (3) a proof of soundness using finite trials compatible with our energy assumption. 
\end{abstract}

\section{Introduction}
Quantum Random Number Generators (QRNG) provide a practical source of randomness since they can be realized with elementary components (such as single photon detectors and lasers diodes) and reach high (Gbit/s) rates \cite{ref:QRNGrev2016,ref:QRNGrev2017}. Furthermore, unlike classical physical generators or pseudo-random algorithms, they rely on quantum processes that are inherently random. This means that even with detailed knowledge of the initial conditions of the QRNG and unlimited computational power, an external adversary can gain no information about the random output. In practice, however, assessing the actual entropy produced by a real QRNG, which may be prone to multiple imperfections and malfunctions, requires a detailed modelling, making the randomness analysis very specific to a given device, difficult to verify, and possibly relying on unwarranted trust in the device components \cite{ref:am2016}.

This weakness of standard QRNG has motivated Device-Independent (DI) QRNG schemes, where the observed violation of a Bell inequality among multiple devices serves as a rigorous certificate that a certain amount of entropy has been produced \cite{ref:pm2013,ref:fgs2013}. This assurance does not rely on any modelling of the devices, which can be treated in a black-box manner, but only requires that the devices satisfy certain causality constraints,  usually that they do not communicate. Though DI QRNGs are conceptually very compelling and have even be demonstrated experimentally \cite{ref:pam2010, ref:rkc2018, ref:lzl2018}, they require challenging loophole-free Bell tests, which precludes real-life implementations with present day technology.

The semi-DI approach aim to retain the conceptual advantages of DI schemes, while making their implementation easier, and in particular avoiding the necessity of using entanglement and loophole-free Bell tests. Their experimental requirements and generation rates are typically similar to standard QRNGs \cite{ref:lb2015, ref:bme2016}, while their theoretical analysis is similar to fully DI schemes as it relies on the observation of certain statistical features akin to the violations of Bell inequalities. However, semi-DI devices cannot be fully treated in a black-box manner, but must satisfy one or a small set of assumptions, such as a bound on the dimension of the relevant Hilbert space \cite{ref:pb2011}. 

In \cite{ref:correnergy}, we introduced a simple semi-DI prepare-and-measure scenario, where the only required assumption is a bound on the average value of a natural physical observable, such as the energy of the prepared states. The device we considered, see Fig.~\ref{fig:pm_setup}, consists of two distinguishable parts: a source (S) and a measurement apparatus (M). The source S prepares one of two quantum systems,  depending of an external control variable $x\in\{1,2\}$, which are then measured at M, yielding a binary outcome $a\in\{\pm 1\}$. The correlations between the output $b$ of the measurement apparatus M and the control variable $x$ of the source S can be quantified by the two quantities $E_1,E_2$, where
\begin{equation}\label{eq:corr}
	E_x=\text{Pr}(a=+1|x)-\text{Pr}(a=-1|x)\,,
\end{equation}
for $x\in\{1,2\}$, indicates how much the output $a$ is biased depending on $x$. 

It is shown in \cite{ref:correnergy} that the observation of certain correlations $\vE = (E_1,E_2)$ between S and M guarantees that the output $a$ is random, similarly to the observation of nonlocal correlations in Bell scenarios. This conclusion is valid assuming only a bound on the average energy (as defined precisely below) of the states emitted by S. But apart from this assumption no other assumptions are made on S or M, in particular the measurement apparatus M can be treated in a fully black-box manner. The scenario is therefore semi-DI. The interest of this proposal is that very simple optical implementations, involving only the preparation of attenuated coherent states and homodyne measurements or single-photon threshold detectors, can produce correlations in the randomness generating regime.
\begin{figure}[t]
	\centering
	\includegraphics[scale=1]{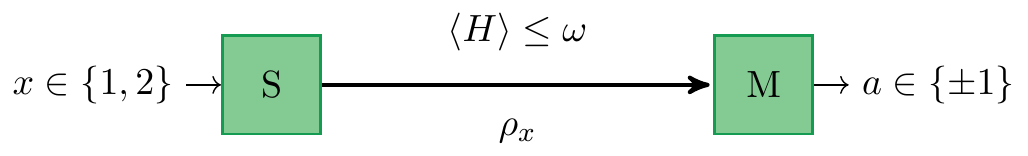}
	\caption{\label{fig:pm_setup}
		Prepare-and-measure scenario considered here involving a device made of a source S and a measurement apparatus M. We assume that the states send by S to M have a bounded average energy.}
\end{figure}

The work \cite{ref:correnergy} showed the existence of inherently random correlations in the energy constrained semi-DI scenario by deriving Bell-type inequalities which are necessarily satisfied by all correlations admitting a deterministic explanation, but which can be violated by quantum correlations. While this immediately suggest the possibility of using such quantum correlations to design semi-DI QRNG protocols, the results of \cite{ref:correnergy} did not tell how much randomness can be certified from such correlations, neither in an ideal situation nor in a real protocol where the correlations are estimated using finite statistics. In the present work, we fill this gap and explicitly introduce and prove the soundness of a semi-DI QRNG protocol based on the ideas of \cite{ref:correnergy}.

Our results are based on three new elements, each developed in their own section.
\begin{itemize}
	\item \textbf{Semidefinite charcteristion of quantum correlations.} First, in Section~\ref{sec:sdp}, we derive a new semidefinite programming (SDP) characterisation of quantum correlations in our semi-DI prepare-and-measure scenario. The set of quantum correlations was already characterized in \cite{ref:correnergy}, but the new SDP characterization is an essential ingredient for determining their randomness in the subsequent section. SDP characterisations also appear in DI scenarios \cite{ref:npa2008,ref:dt2008} and we show that they are closely related by providing an explicit mapping between the quantum set in our scenario and the quantum set in the Bell-CHSH scenario. 
	
	\item \textbf{Quantifying randomness from correlations.} Second, in Section~\ref{sec:entropy}, we provide a simple algorithm to bound the worst case conditional Shannon entropy of the output $a$ as a function of the correlations $\vE$. This quantity yields better rates than the guessing probability, which is usually used in DI QRNG \cite{ref:arv2016,ref:adfrv2018}. Our algorithm uses the previously derived SDP characterisation of the quantum set (and thus could also be adapted to DI QRNG). We illustrate our method by computing optimal asymptotic rates for some of the optical implementations proposed in \cite{ref:correnergy}.
	
	\item \textbf{Protocol for randomness generation.} Finally, in Section~\ref{sec:security_proof}, we provide an explicit semi-DI QRNG protocol and prove its soundness. Our proof relies on the framework of \cite{ref:zkb2018,ref:kzb2017}, which we generalize in two ways. First, we show how to take into account a constraint on the average energy of the states sent over multiple rounds. This type of assumption is technically different from the no-communication assumption in DI QRNG based on standard Bell tests, as the the energy is allowed to fluctuate arbitrarily from round to round provided its average is bounded. Secondly, we base our computation of the entropy rate on the single round conditional Shannon entropy, in the manner of \cite{ref:arv2016,ref:adfrv2018}, instead of the probability estimation factors of \cite{ref:kzb2017}. Tough the two alternatives have the same asymptotic rate, the Shannon entropy is easier to compute using the algorithm of Section~\ref{sec:entropy}. 
\end{itemize}

This paper is somewhat abstract since, in the DI spirit, it applies to any possible implementation of the energy constrained prepare-and-measure scenario that we consider. For the description of explicit implementations and further physical motivation of this scenario, we refer to \cite{ref:correnergy}.


\section[Semidefinite characterization of quantum correlations]{Semidefinite characterization of quantum\\ correlations}
\label{sec:sdp}

\subsection{Definition}
We start by reminding the general scenario considered in \cite{ref:correnergy}. 
As stated in the Introduction, we consider the prepare-and-measure scheme depicted in Figure~\ref{fig:pm_setup}, with $E_x$ defined in eq.~(\ref{eq:corr}) describing the correlations between the output $a\in\{\pm 1\}$ of the measurement apparatus $M$ and the input $x\in \{1,2\}$ of the source S. 
According to quantum theory, we can write in full generality
\begin{equation}
	E_x=\Tr[\rho_x M]\,,
\end{equation}
where $\rho_x$ is the state sent by the source when $x$ is selected and $M$, with $-1\leq M\leq 1$, is an observable describing the measurement process.

In the absence of any additional information, any correlations $\vect{E}=(E_1,E_2)$ are in principle possible between S and M. To constraint further our framework, we introduce a second observable $O$, with the following two properties:
\begin{flalign}
\begin{minipage}{0.8\textwidth}
	\begin{itemize}
		\item $O$ has lowest eigenvalue $0$ and unit gap, i.e., all other eigenvalues 	are greater or equal to $1$;
	\end{itemize}
\end{minipage}&&\label{eq:defineH}\\
\begin{minipage}{0.8\textwidth}
	\begin{itemize}
		\item the eigenspace associated to the lowest eigenvalue 0 is non-degenerate and corresponds to a single eigenvector $\ket{0}$.
	\end{itemize}
\end{minipage}&&\label{eq:defineH2}
\end{flalign}
Examples of such an observable, in the case were the underling Hilbert space is the Fock space of a quantum optical system, would be the particle number operator or the projection on the non-vacuum subspace. Since both such observables are related to the energy content of the quantum optical system, we will refer in the following to any observable $O$ satisfying the properties (\ref{eq:defineH}-\ref{eq:defineH2}) as an \emph{energy observable} and to the nondegenerate lowest eigenstate $|0\rangle$ as the \emph{vacuum state} (without excluding other possible physical realizations for $O$, see, e.g., example 2.3.2 of \cite{ref:correnergy}).

Having introduced the observable $O$, we now make the assumption that the average energies $\Tr[\rho_x O]$ are upper-bounded, i.e. 
\begin{equation}
	\label{eq:ub}
	\Tr[\rho_x O]\leq 	\omega_x\,,
\end{equation}
for $x\in\{1,2\}$. The purpose of introducing this assumption is that it implies fundamental constraints on the possible correlations between S and M, which are valid for any states $\rho_x$ and observable $M$.
This can intuitively be understood by noticing that the energy upper-bounds $\vect{\omega}=(\omega_1,\omega_2)$ are related to how distinguishable the two states $\rho_1$ and $\rho_2$ are. If $\omega_1$ and $\omega_2$ are very small, then the two states $\rho_1$ and $\rho_2$ must both be close to the vacuum state $|0\rangle$ and thus cannot be perfectly distinguished. This in turn implies that $E_1$ and $E_2$ cannot be too different. In the extreme case where $\omega_1=\omega_2=0$, then $\rho_1=\rho_2=|0\rangle\langle 0|$ and therefore necessarily $E_1=E_2$, which means that the observation of $a$ does not carry any information about $x$. Note that, to reach this conclusion, assumption (\ref{eq:defineH2}) on $O$ is essential. If the lowest eigenspace of $O$ is degenerate, then $\rho_1$ and $\rho_2$ can be orthogonal, hence perfectly distinguishable, even when $\omega_1=\omega_2=0$.


Let the tuple $(\vect{E},\vect{\omega})=(E_1,E_2,\omega_1,\omega_2)$ specify a possible pair of values for the correlations $(E_1,E_2)$ between S and M and a possible pair of upper-bounds $(\omega_1,\omega_2)$ on the average energies of the states emitted by S. We refer to such a tuple as a possible \emph{behaviour} of our observed system.
As discussed above, not every behaviour $(\vect{E},\vect{\omega})$ is physically realizable within quantum mechanics because of a tradeoff between the degree of correlations $\vE$ and the energy $\vect{\omega}$ of the emitted states. Formally, we define behaviours admitting a quantum representation as follows.
\begin{definition}
	\label{def:gen_q_set}
	A behaviour $(\vect{E},\vect{\omega})=(E_1,E_2,\omega_1,\omega_2)$ admits a quantum representation if there exist two states $\rho_x$, an observable $-1\leq M\leq 1 $, and an energy operator $O$, satisfying \eqref{eq:defineH} and $\eqref{eq:defineH2}$, such that 
\begin{IEEEeqnarray}{rL}
		\IEEEyesnumber*\label{eq:gen_q_set_b}\IEEEyessubnumber*
		\Tr[\rho_x M]&=E_x\label{eq:gen_q_set}\\
		\Tr[\rho_x O]&\leq \omega_x  \label{eq:gen_q_setb}\,,
	\end{IEEEeqnarray}
or if $(\vect{E},\vect{\omega})$ is a convex combination of behaviours of the above form.

We denote $\setQ \subset \R^4$  the set of all behaviours $(\vect{E},\vect{\omega})$ admitting a quantum representation. 
\end{definition}
Clearly, a valid quantum behaviour must satisfy the trivial constraints $-1\leq E_x\leq 1$ and $\omega_x\geq 0$. But it must also satisfy additional non-trivial constraints. In the next subsection, we give a complete characterization of the set $\setQ$ of quantum behaviours. But first, let us make some general remarks.

\addtocounter{rem}{1}
\paragraph{Remark \arabic{rem}.} We allow for convex combinations in the definition of a quantum behaviour. This means that the source and measurement apparatus can be correlated via some shared randomness $\lambda$. Explicitly, a behaviour $(\vect{E},\vect{\omega})$ is thus quantum if
	\begin{IEEEeqnarray*}{rL}
		\IEEEyesnumber*\IEEEyessubnumber*
		{\sum}_\lambda p_\lambda \Tr[\rho_x^\lambda M^\lambda]&=E_x\label{eq:gen_q_set_hv}\\
		{\sum}_\lambda p_\lambda \Tr[\rho_x^\lambda O^\lambda]&\leq \omega_x \label{eq:gen_q_setb_hv}
	\end{IEEEeqnarray*}
	where $p_\lambda$ is the probability distribution of the shared randomness, and where the states $\rho_x^\lambda$, observable $M^\lambda$, and energy operator $O^\lambda$ can all depend on $\lambda$. This dependency on $\lambda$ can represent, e.g., hidden physical fluctuations that can affect the source and measurement apparatus separately or jointly.

\addtocounter{rem}{1}
\paragraph{Remark \arabic{rem}.} Though we allow S and M to be correlated via shared, classical randomness, we implicitly assume that they do not share entanglement. More generally, if S and M shared prior entanglement, we should write 
$E_x = \Tr[ M_{\mathcal{SM}}\, \$^x_{\mathcal{S}}\otimes 1_{\mathcal{M}}(\rho_{\mathcal{SM}}) ]$ and 
$\omega_x \geq \Tr_S[ O_\mathcal{S} \$^x_{\mathcal{S}}(\rho_{\mathcal{S}})]$, where the indices $\mathcal{S}$, $\mathcal{M}$, refer to the source and measurement apparatus initial systems, respectively, and where $\$^x_{\mathcal{S}}$ denotes a completely positive trace-preserving super-operator acting on the source's systems depending on the value of the control variable $x$. Understanding how this broader setting modify our results is an interesting question, which we do not attempt to resolve here.

\addtocounter{rem}{1}
\paragraph{Remark \arabic{rem}.} Our framework is semi-device-independent in two possible ways. First, we could view our scenario as a prepare-and-measure scenario with two preparation choices, $x=1$ or $x=2$, and two possible measurement $M$ or $O$, so that both the values of $\vE$ and $\vect{\omega}$ are obtained through measurements (this perspective is developed in Appendix~\ref{sec:ent}). While the states $\rho_x$ and the measurement $M$ are completely arbitrary, this is not the case for the measurement $O$ which must satisfy property (\ref{eq:defineH2}) and thus be partly trusted. Thus our scenario is \emph{semi}-device-independent because of this  assumption on the energy operator $O$ (and also because we implicitly assume that S and M do not share prior entanglement). Note that property (\ref{eq:defineH}) simply defines the outcome set associated to the observable $O$ and thus does not really represent an assumption (in the same way that $-1\leq M\leq 1$ follows from the fact that $M$ has binary outcomes). 


But there is also another sense in which our framework is semi-device-independent. As presented in our previous work \cite{ref:correnergy}, in potential applications, it is more natural to see the upper-bounds (\ref{eq:ub}) on $\vect{\omega}$ as a given promise on the quantum systems, rather than information obtained through an actual measurement of $O$. That is, we may have some partial knowledge about the quantum systems produced by the source, which imply bounds of the form (\ref{eq:defineH2}). For instance, we may assume that the source sends optical quantum systems with a large vacuum component. Then the semi-device-independence aspect of our framework refers to an assumption on the source itself rather than to a hypothetical measuring device $O$. This is the approach that we generally take below (except in Appendix~\ref{sec:ent}).

\addtocounter{rem}{1}
\paragraph{Remark \arabic{rem}.}  As we explained above, the purpose of the upper-bounds \eqref{eq:ub} is to constraint the distinguishability of the states $\rho_1$ and $\rho_2$, which in turn implies constraints on $E_1$ and $E_2$. There could of course be other ways to constrain the distinguishability of the emitted states. We choose our approach because it is related to a natural physical property of the emitted states (see the discussion and motivation in \cite{ref:correnergy} for more details), moreover it applies to mixed states, contrarily to the scalar product bound used in \cite{ref:bme2016}).

\subsection[Semidefinite representation of the quantum set Q]{Semidefinite representation of the quantum set $\mathcal{Q}$}
\label{sec:sdprep}

The next proposition provides a useful characterization of the quantum set $\setQ$ in the form of a positive semidefinite (SDP) constraint. This characterization will later be shown to be equivalent to the formula given in \cite{ref:correnergy}.

\begin{theorem}
	\label{prop:sdp_char}
	A behaviour $(\vE,\vect{\omega})$ is in $\setQ$ if and only if there exist real numbers $u,v,\eta_1,\eta_2$ with $\eta_1\leq \omega_1$ and $\eta_2\leq \omega_2$ such that the symmetric matrix
	\begin{IEEEeqnarray}{rL}
		\label{eq:sdp_char}		
		\Gamma=\left(\begin{IEEEeqnarraybox}[][c]{c?c?c?c}
			1 & u & E_1 & 2 \eta_1-1\\
			  & 1 &	E_2 & 2 \eta_2-1\\
			  &   & 1      & v\\
			  &   &        & 1
		\end{IEEEeqnarraybox}\right) \succeq 0	   
	\end{IEEEeqnarray}
is positive semidefinite.
\end{theorem}
The necessary part of this statement is a direct consequence of the following Lemma.
\begin{lemma}\label{lemma1}
	Any extremal behaviour $(\vE,\vect{\omega})$ of the quantum set $\setQ$ admits a representation (\ref{eq:gen_q_set}),(\ref{eq:gen_q_setb}) which is two-dimensional and of the form
	\begin{IEEEeqnarray*}{rL}
		\IEEEyesnumber*
		\label{eq:bloch_mapping}
		\IEEEyessubnumber*
		\rho_x &=\frac{1}{2}\left(\id + \vect{n}_x \cdot \vsigma \right)\label{eq:bloch_mapping_r}\\
		M&=\vect{m}\cdot \vsigma\label{eq:bloch_mapping_m}\\
		O &= \frac{1}{2}\left(\id+\vect{k}\cdot\vsigma\right)\label{eq:bloch_mapping_h}
	\end{IEEEeqnarray*} 
	where $\vect{n}_x$, $\vect{m}$, and $\vect{k}$ are unit vectors in $\mathbb{R}^3$ and $\vsigma=(\sigma_x,\sigma_y,\sigma_z)$ is the vector of Pauli matrices.
\end{lemma}
This lemma can also be used to map behaviours in our scenario to the behaviours of regular Bell tests, see Appendix~\ref{sec:ent}.
\begin{proof}[Proof of Lemma~\ref{lemma1}]
Consider an extremal behaviour $(\vE,\vect{\omega})$ in $\mathcal{Q}$ and let $\rho_1,\rho_2,M,O$ be a quantum representation for it as in Definition~1. Since $(\vE,\vect{\omega})$ is extremal, we can assume that the states $\rho_x$ are pure, i.e, $\rho_x=|\rho_x\rangle\langle\rho_x|$. (Indeed, behaviours having a mixed-state representation can be seen as convex combination of pure-state behaviours). The two pure states $|\rho_1\rangle$, $|\rho_2\rangle$ span a two-dimensional subspace $\mathcal{V}$ of the entire Hilbert space $\mathcal{H}$ and can thus be written in that subspace as in eq.~(\ref{eq:bloch_mapping_r})
for some unit vectors $\vect{n}_x$. 

The correlations $E_x$ obviously only depend on the projection of $M$ on that two-dimensional subspace $\mathcal{V}$. This projection can be written in full generality as
\begin{equation}
M =m_0\,\id + \vect{\tilde m} \cdot \vsigma\,,
\end{equation}
where $|m_0|+|\vect{\tilde m}|\leq 1$ (and where by a slight abuse of notation, we use the same symbol for $M$ and its projection on $\mathcal{V}$). The observable $M$ can be interpreted  as a convex combination of three simpler measurements since we can write
\begin{equation}\label{eq:dec_M}
M=p_1\,\id +p_2\, (-\id) + p_3\, \vect{m}\cdot \vsigma\,,
\end{equation}
where $p_1=(1+m_0-|\vect{\tilde m}|)/2$, $p_2=(1-m_0-|\vect{\tilde m}|)/2$, $p_3=|\vect{\tilde m}|$, $\vect{m}=\frac{\vect{\tilde m}}{|\vect{\tilde m}|}$, and where the condition $|m_0|+|\vect{\tilde m}|\leq 1$ guarantees that $p_1,p_2\geq 0$. Since we consider an extremal behaviour, we deduce that $M$ is either of the form $M=\pm \id$ or of the form $M=\vect{m}\cdot \sigma$. 

The case $M=\pm \id$ corresponds to a measurement that yields a deterministic outcome that does not depend on $x$ : $E_1=E_2=\pm1$. But such correlations can always be simulated by sending the vacuum state of $O$, i.e., $\ket{\rho_1}=\ket{\rho_2}=\ket{0}$, and using the measurement $M=\pm (\ket{0}\bra{0}-\ket{1}\bra{1})=\pm\vect{1}_z\cdot \sigma$. This represents a valid quantum realization since the mean energies of the states, being $0$, are lower than $\omega_x$ in accordance with the bound (\ref{eq:gen_q_setb}).  Furthermore, it is of the form of eqs.~(\ref{eq:bloch_mapping}) with $\vect{n}_x=\vect{1}_z$, $\vect{m}=\pm\vect{1}_z$, $\vect{k}=-\vect{1}_z$, and thus the statement of the Lemma is verified in that case. 

For the case $M=\vect{m}\cdot \sigma$, it only remains to show (\ref{eq:bloch_mapping_h}), since $\rho_x$ and $M$ have already been shown to be of the form (\ref{eq:bloch_mapping_r}) and (\ref{eq:bloch_mapping_m}), respectively. Let $\tilde O=\id-|\tilde 0\rangle\langle \tilde 0|$, where $|\tilde 0\rangle$ is the normalized projection of $|0\rangle$ (the vacuum state of $O$) on the two-dimensional space $\mathcal{V}$ (if the projection of $|0\rangle$ in $\mathcal{V}$ is zero, simply define $|\tilde 0\rangle$ as an arbitrary state in $\mathcal{V}$). It is easily verified that $\langle\rho_x|O|\rho_x\rangle\geq \langle\rho_x|\left(\id - |0\rangle\langle 0|\right)|\rho_x\rangle\geq \langle\rho_x|	\left(\id - |\tilde 0\rangle\langle \tilde 0|\right)|\rho_x\rangle=\langle\rho_x|\tilde O|\rho_x\rangle$. Thus, $\omega_x\geq \langle\rho_x|O|\rho_x\rangle\geq \langle\rho_x|\tilde O|\rho_x\rangle$, and we can replace the operator $O$ by $\tilde O$ while still having a proper quantum realization. The energy operator $\tilde O$ is now a rank-one projector in the two-dimensional space $\mathcal{V}$ and can thus be written as in (\ref{eq:bloch_mapping_h}).	
\end{proof}

\begin{proof}[Proof of Theorem~\ref{prop:sdp_char}]
At first, we consider an extremal behaviour $(\vE,\vect{\omega})$ and consider the Gram matrix $\Gamma$ of the unit vectors $\vect{n}_1$, $\vect{n}_2$, $\vect{m}$, $\vect{k}$, which are defined in Lemma~1. Then using the inner product between Block vectors, we find
\begin{equation}
\Gamma=\left(\begin{IEEEeqnarraybox}[][c]{c?c?c?c}
			1 & \vn_1\cdot\vn_2 & \vn_1\cdot\vm & \vn_1\cdot \vect{k}\\
			  & 1 				& \vn_2\cdot\vm & \vn_2\cdot \vect{k}\\
			  &   				& 1      		& \vm\cdot \vect{k}\\
			  &   				&       		& 1
		\end{IEEEeqnarraybox}\right)=\left(\begin{IEEEeqnarraybox}[][c]{c?c?c?c}
			1 & u & E_1 & 2\eta_1-1\\
			  & 1 				& E_2 & 2\eta_2-1\\
			  &   				& 1      		& v\\
			  &   				&       		& 1
		\end{IEEEeqnarraybox}\right)
\end{equation}
where $u,v,\eta_1\leq \omega_1, \eta_2\leq \omega_2$ are real numbers. By construction $\Gamma$ is semi-positive definite. Thus for any extremal behaviour $(\vE,\vect{\omega})$ belonging to the set $\mathcal{Q}$, there exists a matrix $\Gamma$ satisfying the SDP constraint (\ref{eq:sdp_char}). As this is a convex constraint in the variables $(\vE,\vect{\omega})$ \cite{ref:boyd}, the SDP constraint must hold also when considering convex combination of extremal behaviours, i.e., it must hold for all behaviours in $\setQ$. This establishes the necessary condition in the Theorem.

To prove sufficiency, consider a semi-definite positive matrix $\Gamma$, as in \eqref{eq:sdp_char} with unit diagonal elements. Then there must exist four 4-dimensional unit vectors $\vn_1,\vn_2,\vm, \vect{k}$ such that $E_x=\vn_x \cdot \vm$ and $2 \eta_x-1=\vn_x \cdot \vect{k}$. We can project $\vm \mapsto \vm'$ (with $|\vm'|\leq 1$) onto the space spanned by $\vn_1,\vn_2,\vect{k}$, as this leaves $E_x=\vn_x \cdot \vm$ and $2 \eta_x-1=\vn_x \cdot \vect{k}$ invariant. Since this is a three dimensional subspace, this defines four Bloch vectors that can be mapped through \eqref{eq:bloch_mapping} to a valid quantum representation (with the only difference that $|\vm'|\leq 1$)\footnote{%
	The fact that $|\vm'|\leq 1$ in the sufficiency part of the proof originates from the fact that the starting behaviour $(\vE,\vect{\omega})$ is not necessarily extremal. Given an arbitrary measurement $M'=\vm'\cdot \vsigma$, with $|\vm'|\leq 1$, we can, however, rewrite it in the form \eqref{eq:dec_M} with $p_1=p_2=(1-|\vm'|)/2$, $p_3=|\vm'|$, $\vm=\vm'/|\vm'|$, and thus interpret it, as below eq.~\eqref{eq:dec_M}, as a convex combination of three strategies with measurements given by unit Bloch vectors.}. This proves that the SDP constraint \eqref{eq:sdp_char} is tight.%
\end{proof}

An immediate consequence of the construction at the end of the proof of Theorem~\ref{prop:sdp_char} is that any behaviour $(\vE,\vect{\omega})\in\setQ$ can be realized without shared randomness:

\begin{corr}
	Any point in the set $\setQ$ can be realized with a two-dimensional representation of the form of eqs.~(\ref{eq:bloch_mapping}), where $\vect{n}_x$ and $\vect{\omega}$ are unit vectors and $|\vect{m}|\leq 1$.	 
\end{corr}

\subsection{Closed-form representations}

The SDP characterization in Theorem~\ref{prop:sdp_char} is very convenient numerically, but not so easy to picture geometrically. We now provide two alternative closed-form characterizations of the quantum set~$\setQ$. First of all, let us remark that the energies $\omega_1$, $\omega_2$ must be sufficiently small -- specifically their sum $\omega_1+\omega_2$ must be bounded by one -- to get non-trivial constraints on the set of quantum behaviours.

\begin{prop}
	\label{prop:trivial_cons}
	Let $(\vE,\vect{\omega})$ be a behaviour satisfying the trivial constraints $-1\leq E_x\leq 1$ and $\omega_x\geq 0$ and the condition $\omega_1+\omega_2\geq 1$. Then it belongs to the quantum set $\mathcal{Q}$.		
\end{prop}

\begin{proof}
Though this can be established from the semidefinite positivity of the matrix $\Gamma$ in \eqref{eq:sdp_char}, it also follows directly from the definition \eqref{eq:gen_q_set_b}. The necessity of the trivial constraints $-1\leq E_x\leq 1$ and $\omega_x\geq 0$ is obvious. If in addition $\omega_1+\omega_2\geq 1$, one can define $\omega'_1\leq \omega_1$, $\omega'_2\leq \omega_2$ with $\omega'_1+\omega'_2=1$. Let $O=I-|0\rangle\langle 0|$. Then the two states $|\rho_1\rangle = \sqrt{1-\omega'_1}|0\rangle +\sqrt{\omega'_1}|1\rangle$ and  
$|\rho_2\rangle = \sqrt{1-\omega'_2}|0\rangle -\sqrt{\omega'_2}|1\rangle$ have mean energies $\omega'_1$ and $\omega'_2$, in accordance with (\ref{eq:gen_q_setb}). Furthermore, they are orthogonal -- thus perfectly distinguishable -- and we can reproduce as in \eqref{eq:gen_q_set} any correlations $(E_1,E_2)\in[-1,1]^2$ using the observable $M=E_1 |\rho_1\rangle\langle \rho_1|+E_2 |\rho_2\rangle\langle \rho_2|$. 
\end{proof}

The following corollary to Theorem~\ref{prop:sdp_char} provides two closed-form expressions, \eqref{eq:trig_form} and (\ref{eq:scal_form}), for the non-trivial constraints characterizing $\setQ$ when $\omega_1+\omega_2<1$. The formula \eqref{eq:scal_form} was already proven in \cite{ref:correnergy}, but we rederive it here. Formula (\ref{eq:trig_form}) is new. It is interesting to compare it to the linear inequality characterizing the classical set in our scenario, which was defined and derived in \cite{ref:correnergy}:
\begin{IEEEeqnarray}{rL}
	\label{eq:classical_bound1}
	|E_1 - E_2| \leq 2(\omega_1 + \omega_2)\,.
\end{IEEEeqnarray}

\begin{corr}\label{th:closedform}
Let $(\vE,\vect{\omega})$ be a behaviour satisfying the trivial constraints $-1\leq E_x\leq 1$ and $\omega_x\geq 0$ and the condition $\omega_1+\omega_2\leq 1$. Then it belongs to the quantum set $\mathcal{Q}$ if and only if  
\begin{IEEEeqnarray}{rL}
	\label{eq:trig_form}
		|\asin E_1-\asin E_2| 
			&\leq 2(\asin \sqrt{\omega_1} + \asin\sqrt{\omega_2})\,,
	\end{IEEEeqnarray}
or, equivalently, if and only if 
\begin{IEEEeqnarray}{rL}
	\label{eq:scal_form}
	\tfrac{1}{2}\left( \sqrt{1+E_1}\sqrt{1+E_2} +\sqrt{1-E_1}\sqrt{1-E_2}\right)\geq \sqrt{1-\omega_1}\sqrt{1-\omega_2}-\sqrt{\omega_1}\sqrt{\omega_2}\,.
	\end{IEEEeqnarray} 
\end{corr}

\begin{proof}
Define $\theta_x = \asin E_x$ and $\mu_x = \asin (2\omega_x -1)$. All these angles are well-defined since $-1\leq E_x\leq 1$ and $0\leq \omega_x\leq 1$, where the last condition follows from $\omega_1+\omega_2<1$.

Now assume in addition that $(\vE,\vect{\omega})\in\setQ$, i.e., that $(\vE,\vect{\omega})$ satisfy Theorem~\ref{prop:sdp_char}. Lemma 13 in \cite{ref:npa2008} gives condition for a matrix of the form (\ref{eq:sdp_char}) to be positive semidefinite. Using this Lemma and setting $\epsilon_x = \asin (2\omega_x - 1)$, we find that Theorem~\ref{prop:sdp_char} is equivalent to the statement that there exist $\epsilon_x$ such that
\begin{equation}\label{eq:epsmu}
-\pi/2\leq \epsilon_x\leq \mu_x
\end{equation}
and
\begin{IEEEeqnarray*}{rL}
	\IEEEyesnumber \label{eq:sin_cond} \IEEEyessubnumber*
	|\theta_1 - \theta_2| + |\epsilon_1 + \epsilon_2 | &\leq \pi\\
	|\theta_1 + \theta_2| + |\epsilon_1 - \epsilon_2 | &\leq \pi\,.
\end{IEEEeqnarray*}
For given $\theta_1,\theta_2$, the set $R_1$ of couples $(\epsilon_1,\epsilon_2)$ satisfying Equation~\eqref{eq:epsmu} is a rectangle depicted in Figure~\ref{fig:intersection}. A straightforward calculation shows that the set $R_2$ of couples $(\epsilon_1,\epsilon_2)$ satisfying \eqref{eq:sin_cond} is also a rectangle with four corners of the form $\pm (a,b), \pm (b,a)$ with $a \geq \pi/2$. From the Figure, it is clear that the two rectangles have a non-empty intersection, i.e., that Theorem~\ref{prop:sdp_char} is satisfied, if and only if 
\begin{IEEEeqnarray}{rL}
	\label{eq:trig_form_alt}
	|\theta_1 - \theta_2 | - (\mu_1 + \mu_2) \leq \pi.
\end{IEEEeqnarray}
Using elementary trigonometric relations, this last condition can be rewritten as \eqref{eq:trig_form}.

\begin{figure}
	\centering
	\includegraphics[scale=1]{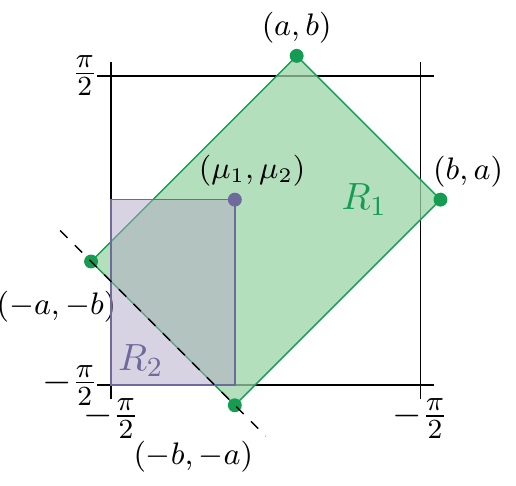}
	\caption{ 
	\label{fig:intersection} The rectangles $R_1$ and $R_2$ have a non-empty intersection if and only if the upper-right corner $(\mu_1,\mu_2)$ of $R_1$ satisfies the lower-left facet inequality $\epsilon_1+\epsilon_2\geq |\theta_1 - \theta_2 | - \pi$ of the rectangle $R_2$.
	}
\end{figure}

Using some trigonometric manipulations, we now show that the formula \eqref{eq:scal_form} is equivalent to \eqref{eq:trig_form}. Define $\alpha_x,\beta_x \in [0,\pi/2]$ such that $\sin \alpha_x = \sqrt{\frac{1+E_x}{2}}$ and $\sin\beta_x = \sqrt{\omega_x}$. Then \eqref{eq:scal_form} can be rewritten as $\cos(\alpha_1 - \alpha_2) \geq \cos(\beta_1 + \beta_2)$ which is equivalent to $|2\alpha_1-2\alpha_2| \leq 2(\beta_1 + \beta_2)$ in the range $-\pi/2\leq \alpha_1-\alpha_2\leq \pi/2$ and $0\leq \beta_1+\beta_2\leq \pi$. Writing $|2\alpha_1-2\alpha_2|=|2\alpha_1-\pi/2-(2\alpha_2-\pi/2)|$ and using the relations $\pm (2\alpha_x-\pi/2) = \pm \asin E_x$ and $\beta_x = \asin\sqrt{\omega_x}$, we find that this is equivalent to equation \eqref{eq:trig_form}, as claimed.
\end{proof}

Remark that in the proof of Corollary~\ref{th:closedform}, we actually only use the condition $\omega_1+\omega_2\leq 1$ to derive the weaker conditions $\omega_x\leq 1$. The statement of the Corollary is thus also true assuming only these weaker conditions, with eqs. (\ref{eq:trig_form}) and (\ref{eq:scal_form}) being trivially satisfied when $\omega_1+\omega_2\geq 1$.


\section{Quantifying randomness from correlations}
\label{sec:entropy}

It was shown in \cite{ref:correnergy} that there are certain  quantum behaviours $(\vE,\vect{\omega}) \in \setQ$ which exhibit genuine quantum randomness in the sense that the output $a$ of the device cannot be perfectly predicted whatever the underlying quantum representation giving rise to this behaviour. This was shown by deriving inequalities which are necessarily satisfied by any behaviours admitting a deterministic representation and then finding quantum behaviours $(\vE,\vect{\omega})$ which violate these inequalities. This result is similar in spirit to the violation of Bell inequalities by separated no-signalling devices, which witness genuine randomness independently of the devices' implementation. As a matter of fact, the violation of a Bell inequality only implies the presence of a non-zero amount of randomness. But it is also possible to obtain quantitative lower-bounds on the amount of randomness compatible with given non-local correlations \cite{ref:pam2010, ref:nps2014, ref:bss2014}. Similarly, we show in this section how to obtain quantitative bounds on the amount of randomness compatible with a given behaviour in the scenario that we consider here.

\subsection{Formulation of the problem}
We start by defining precisely what we mean by `randomness', how we measure it, and what is the problem we aim to solve.

Assume that we hold a prepare-and-measure device, as defined in the previous section, choose an input $x\in\{1,2\}$ according to a known probability distribution $p(x)$, enter $x$ in the device, and obtain the output $a\in\{\pm 1\}$. We do not make any detailed assumptions about how the devices operate internally to give rise to the output $a$ -- apart from the fact that $i)$ it should arise from a valid quantum representation  and $ii)$ be compatible with certain energy constraints defined further below.

We are interested in quantifying how random the output $a$ is from the point of view of a hypothetical adversary who, unlike us, could have a detailed physical description of the device. We represent by the symbol $\lambda$ the collection of physical parameters which determine the  behaviour of the device from the adversary's point of view. These parameters may themselves fluctuate randomly, and thus are described by a probability distribution $p(\lambda)$, that is unknown to us. It could for instance be that the randomness that we observe is entirely due to statistical fluctuations of $\lambda$ and that the output of the device is completely deterministic for an adversary happening to known the precise value of $\lambda$.

From the point of view of the adversary, the behaviour of the device is thus characterized by an ensemble $\{p(\lambda), (\vE^\lambda,\vect{\omega}^\lambda) \in \setQ\}$ of behaviours.  The correlations $\vE$ characterizing the output $a$ as seen from the external point of view of the user which has not access to the internal description of the device are then given by
\begin{equation}
	\sum_\lambda p(\lambda) \vE^\lambda=\vE\,.
\end{equation}
We make the following two assumptions regarding the distribution of energies $\vect{\omega}^\lambda$:
\begin{IEEEeqnarray}{rrL}
	\label{eq:maxav}
	&\sum_\lambda\,p(\lambda)\, \vect{\omega}^\lambda &\leq \vect{\omega}_{\text{avg}}\,,
\end{IEEEeqnarray}
and 
\begin{IEEEeqnarray}{rL}
	\label{eq:maxpeak}
	\vect{\omega}^\lambda &\leq \vect{\omega}_{\text{pk}} \fa \lambda\,.
\end{IEEEeqnarray}
The first assumption, which we call the \emph{max-average} assumption, states that there exists an upper-bound on the average energy, where the average is taken over the hidden parameters. The second one, which we call the \emph{max-peak} assumption, states that there is in addition an absolute bound on the energy satisfied for all individual values of the hidden parameters. Of course, this later assumption requires some additional trust in the devices. Both assumptions were already introduced and motivated in \cite{ref:correnergy}. The max-average constraint is particularly interesting from the point of view of semi-device-independent randomness certification since, contrarily to the max-peak assumption, it does not constraint the behaviour of the device at the individual hidden level, but only on average. In particular, it is conceivable to verify it experimentally by measuring the average energy of the states sent over the channel over a large number of rounds. We refer to \cite{ref:correnergy} for a further discussion of these two assumptions.

Note that in the following we do not necessarily need to impose both the max-average and max-peak assumptions, but possibly only one of them. The case where the upper-bound $\vect{\omega}_{\text{avg}}$ on the average energy satisfies $\vect{\omega}_{\text{avg}}\geq \vect{\omega}_{\text{pk}}$ effectively means that one is considering only the max-peak assumption, since the average energy is always bounded by the peak value: $\sum_\lambda\,p(\lambda)\, \vect{\omega}^\lambda \leq \vect{\omega}_{\text{pk}}$. The case where the upper-bound $\vect{\omega}_{\text{pk}}$ on the peak energy satisfies $\vect{\omega}_{\text{pk}}=(\omega_{\text{pk},1},\omega_{\text{pk},2})\geq (1,1)=\vect{1}$ effectively means that one is considering only the max-average assumption since, as follows from Proposition~\ref{prop:trivial_cons}, there are no constraints on the quantum correlations $\vE^\lambda$ when one increase the energy $\vect{\omega}^\lambda$ beyond the value $\vect{1}$. Without loss of generality, we thus always assume in the following that
\begin{equation}\label{eq:ass_energies}
\vect{0}\leq\vect{\omega}_{\text{avg}}\leq \vect{\omega}_{\text{pk}}\leq \vect{1}\,,
\end{equation}
where $\vect{\omega}_{\text{avg}}= \vect{\omega}_{\text{pk}}$ means that we consider only the max-peak assumption (no constraint on the average energy) and $\vect{\omega}_{\text{pk}}=\vect{1}$ means that we consider only the max-average assumption (no constraint on the peak energy).

In the following, we assume that $\vE$, $\vect{\omega}_{\text{avg}}$, and $\vect{\omega}_{\text{pk}}$ are known and given and we seek to find out how random the output $a$ is, from the point of view of the adversary. Note that in a real randomness generation protocol, as considered in the next section, we would estimate the correlations $\vE$ by probing sufficiently many times the device. However, in the present section we assume that we know beforehand this information as our aim for now is simply to understand at a fundamental level, given a certain observed behaviour of the device, how random the output $a$ is.

\subsection{Randomness measures}
Given $\lambda$ and assuming the adversary is also given the input $x$, the output $a$ arises from his point of view with probability $p(a|x,\lambda)=\tfrac{1}{2}(1+aE^\lambda_x)$. The randomness associated with this situation, averaged over the possible values of $\lambda$ and $x$, can be characterized using different quantities. One possibility is the conditional Shannon entropy \cite{cover_elements_2006} $H(A|X,\Lambda) = -{\sum}_{a,x,\lambda} p(a,x,\lambda) \log_2 p(b|x,\lambda)$. If the inputs are chosen independently of the devices so that $p(x,\lambda)=p(x)p(\lambda)$, it can be rewritten as
\begin{IEEEeqnarray}{rL}
	H(A|X,\Lambda) &= \sum_\lambda p(\lambda) H(\vE^\lambda),
\end{IEEEeqnarray}
where
\begin{equation}
H(\vE) = - \sum_{a,x} p(x) p(a|x) \log p(a|x)=- \sum_{a,x} p(x)\frac{1+aE_x}{2}  \log \frac{1+aE_x}{2}\,.
\end{equation}
Note that $H(\vE)$ depends not only on $\vE$, but also on $p(x)$ (but to simplify the notation and because we assume $p(x)$ to be fixed, we do not explicitly indicates this dependence in the notation $H(\vE)$). 

Another possibility is to use the guessing probability \cite{tomamichel_quantum_2016}
\begin{IEEEeqnarray}{rL}
	G(A|X,\Lambda)&=\sum_{\lambda}p(\lambda)G(\vE^\lambda)\,,
\end{IEEEeqnarray}
where $G(\vE^\lambda)=\sum_x p(x) \max_b p(a|x,\lambda)$. The guessing probability can be used to define the min-entropy $H_{\min}(A|X,\Lambda)=-\log_2 G(A|X,\Lambda)$, which lower-bounds the conditional entropy: $H(A|X,\Lambda)\geq H_{\min}(A|X,\Lambda)$.

To obtain the best lower-bound on the device's randomness that is valid independently of the adversary's knowledge and of the device implementation, we must actually optimise the above measures of randomness over all possible ensembles $\{p(\lambda),\vE^\lambda,\vect{\omega}^\lambda\}$ of hidden behaviours compatible with the constraints. For instance in the case of the conditional entropy, we have $H(A|X,\Lambda)\geq H^\star$, with
\begin{IEEEeqnarray*}{R'r'L}
	\label{eq:opti_entropy}
	\IEEEyesnumber
	\IEEEyessubnumber*
	H^\star=
		&\min_{\{p(\lambda),\vE^\lambda,\vect{\omega}^\lambda \}}
			& {\sum}_\lambda p(\lambda) H(\vE^\lambda) \\
		&\text{subject to } 
			& {\sum}_\lambda p(\lambda) \vE^\lambda = \vE\\
			&& {\sum}_\lambda p(\lambda) \vect{\omega}^\lambda = \vect{\omega}_{\text{avg}} \label{eq:opti_max_av_cons}\\
			&& (\vE^\lambda, \vect{\omega}^\lambda)\in \setQ_{\vect{\omega}_{\text{pk}}}\, ,
\end{IEEEeqnarray*}
where $p(\lambda)$ is a valid probability distribution and where
\begin{equation}\label{eq:qpeak}
\mathcal{Q}_{\vect{\omega}_{\text{pk}}}=\{(\vE,\vect{\omega})\in\mathcal Q\text{ with }\vect{\omega}\leq \vect{\omega}_{\text{pk}}\}\,.
\end{equation} 
is the set of behaviours satisfying the energy constraint $\vect{\omega}\leq \vect{\omega}_{\text{pk}}$. The value of $H^\star$ implicitly depends on the correlations $\vE$, the energy assumptions $\vect{\omega}_{\text{avg}},\vect{\omega}_{\text{pk}}$ and the choice of input distribution $p_X$. A similar upper-bound $G(A|X,\Lambda)\leq G^\star$ on the guessing probability  (corresponding to a lower-bound on the min-entropy) can also be obtained by solving the corresponding optimisation problem $G^\star = \max_{\{\{p(\lambda),\vE^\lambda,\vect{\omega}^\lambda \}}{\sum}_\lambda p(\lambda) G(\vE^\lambda)$.

Note that without loss of generality, we have used an equality sign in the constraint (\ref{eq:opti_max_av_cons}) instead of an inequality sign as in (\ref{eq:maxav}), because if there exists a solution with average energy strictly lower than $\vect{\omega}_{\text{avg}}$, one can always increase the energies $\vect{\omega}^\lambda$ to make it exactly equal to $\vect{\omega}_{\text{avg}}$. Note further that if $\vect{\omega}_{\text{avg}}=\vect{\omega}_{\text{pk}}$, i.e., there is no max-average assumption, then one can remove the constraint (\ref{eq:opti_max_av_cons}) in the above formulation.

We show in the following subsection how to solve numerically the above optimization problem using the characterization of the quantum set obtained in Subsection~\ref{sec:sdprep}.

\addtocounter{rem}{1}
\paragraph{Remark \arabic{rem}.} The subset of $(\vE,\vomega) \in \setQ$ for which the lower-bound is zero, i.e. $H^\star = 0$ or $-\log_2 G^\star= 0$, is the subset of behaviours that are not useful from an RNG perspective because there exists a deterministic model that satisfies the assumptions on the energy and reproduces the observed statistics. Indeed in the optimisation problem \eqref{eq:opti_entropy}, if $H^\star = 0$ then $H(A|X,\Lambda=\lambda) = 0$ for all $\lambda$, because $H(A|X,\Lambda=\lambda)\geq 0$, and when $p(x) > 0$ for all $x$, this implies that $E_x^\lambda = \pm 1$ for all $x$ and $\lambda$. Behaviours that admit such a decomposition were characterized in \cite{ref:correnergy}, where they were called classical correlations. When using only a max-average assumption, a behaviour is classical if and only if
\begin{IEEEeqnarray}{rL}
	\label{eq:classical_bound}
	|E_1 - E_2| \leq 2(\omega_{\mathrm{avg},1} + \omega_{\mathrm{avg},2})\,.
\end{IEEEeqnarray}
When using solely a max-peak assumption, a behaviour is classical if and only if $E_1 = E_2$ (in the non-trivial zone  $\omega_{\text{pk},1} + \omega_{\text{pk},2}< 1$). In the case where the input distribution is maximally biased towards a specific value $x = x_0$, we have $H^\star = 0$ if and only if the output is deterministic for that input $x_0$: $E_{x_0}^\lambda = \pm 1$ (but possibly $E_{x}^\lambda \neq \pm 1$ for $x\neq x_0$). The set of behaviours with this property was also characterized in \cite{ref:correnergy}.

For behaviours that are outside the classical sets identified in \cite{ref:correnergy}, the conditional entropy $H^\star$ or the min-entropy $-\log_2 G^\star$ thus take positive values, which can be determined by solving the corresponding optimisation problem.  


\addtocounter{rem}{1}
\paragraph{Remark \arabic{rem}.} We have here implicitly focused on a situation where the adversary has only classical-side information about the output $a$ of the device, represented by the variables $\lambda$. But one could also consider a more general situation where the adversary holds quantum-side information, so he holds a quantum system that is entangled to the device. This would not require a modification of the definition of the guessing probability $G(A|X,\Lambda)$, as shown in \cite{konig_operational_2009}, but would require considering the conditional von Neumann entropy $S(A|X,\Lambda)$ instead of the Shannon entropy $H(A|X,\Lambda)$. However, a (non-optimal) lower bound on the quantum conditional entropy can be obtained from the guessing probability, $S(A|X,\Lambda)\geq -\log_2 G(A|X,\Lambda)$, i.e., with the techniques discussed below.

\subsection[Algorithm for finding entropy bounds through semidefinite programming]{Algorithm for finding entropy bounds through\\ semidefinite programming}

\label{subsec:algo_entropy}

To solve the optimization problem \eqref{eq:opti_entropy}, we provide an algorithm that computes a converging series of lower-bonds $H^\star_{k} \leq H^\star$, for $k \in \N$, with $\lim_{k \rightarrow \infty} H^\star_k = H^\star$, using semidefinite programming. This algorithm relies on three essential elements developed below: (a) the dual formulation of \eqref{eq:opti_entropy}, (b) a linearization of the entropy function $H(\vE)$, and (c) a way to optimize linear constraints over $\setQ$ as a SDP.

Note that the optimisation problem \eqref{eq:opti_entropy} and the alternative version with the guessing probability are part of a more general class of optimisation problems, closely related to convex-roof extensions in entanglement theory. We present some of their properties in Appendix~\ref{sec:apprendix}, in particular their dual formulation.

\paragraph{Dual formulation.} Consider the following dual formulation \cite{ref:boyd} of \eqref{eq:opti_entropy}
\begin{IEEEeqnarray}{R'r'L}
	\label{eq:opti_dual}
		\IEEEyesnumber
		\IEEEyessubnumber*
		H^\star_{dual} =& 
			\sup_{\{\alpha, \vect \beta, \vect \gamma \}}
				& \alpha + \vect \beta \cdot \vect E + \vect \gamma\cdot \vomega_{\avg}\label{eq:opti_dual_obj}\\
			&\st 
				& \alpha + \vect \beta \cdot \vect E + \vect \gamma\cdot \vomega \leq H(\vE) \fa (\vE,\vomega) \in  \setQ_{\vomega_{\pk}} \label{eq:opti_dual_cons}
\end{IEEEeqnarray}
It is easy to see that feasible points of the problem \eqref{eq:opti_dual} provide lower-bounds on $H^\star$. Moreover, it is shown in Appendix~\ref{sec:apprendix} that strong duality holds for problems of this form: $H^\star_{dual} = H^\star$ (and furthermore that the  supremum is actually a maximum when $(\vE,\vomega_{avg}) \in \interior(\setQ)$ is not on the border of the quantum set). Due to the structure of the problem, we also have the following simplification that will be useful later.
\begin{lemma}
	In the dual problem \eqref{eq:opti_dual}, we can restrict to $\vect \gamma \preceq 0$ in general and we can put $\vect \gamma = \vect 0$ when we do not use a max-average assumption (i.e.\ $\vomega_\avg = \vomega_\pk$).
\end{lemma}
\begin{proof}
	The first statement can be seen as follows. Assume that we have an optimal solution of \eqref{eq:opti_dual} with $ \gamma_1 > 0$ (a similar argument holds for $\gamma_2> 0)$. Consider some arbitrary $(\vE,\vomega) \in \setQ_{\omega_{\pk}}$. We then have $\alpha + \vect \beta\cdot \vect E + \gamma_1\omega_{1} + \gamma_2 \omega_2 \leq \alpha + \vect \beta\cdot \vect E + \gamma_1\omega_{\pk,1} + \gamma_2 \omega_2$ since $\gamma_1>0$ and $\omega_1\leq \omega_{\pk}$. Furthermore, it follows from the definition~\ref{def:gen_q_set} of the quantum set and the definition~(\ref{eq:qpeak}) of $\setQ_{\omega_{\pk}}$ that the behaviour $(\vE,\omega_{\pk,1},\omega_2)$ belongs to $\setQ_{\vect{\omega}_{\pk}}$ and thus that $\alpha + \vect \beta\cdot \vect E + \gamma_1\omega_{\pk,1} + \gamma_2 \omega_2\leq H(\vE)$. If we define $\tilde \alpha = \alpha + \gamma_1 \omega_{\pk,1}$ and $\tilde{\vect \gamma} = (0,\gamma_2)$, we thus have shown that $\tilde \alpha + \vect \beta\cdot \vect E +\tilde{\vect \gamma}\cdot \vomega \leq H(E)$ for all $(\vE,\vomega) \in \setQ_{\omega_{\pk}}$, i.e., we have defined a new feasible solution $(\tilde\alpha,\vect{\beta},\tilde{\vect{\gamma}}	)$ satisfying the dual constraint \eqref{eq:opti_dual_cons} and such that $\tilde\gamma_1=0$. Furthermore it achieves a higher value of the objective function because 
	$\tilde \alpha + \vect \beta \cdot \vE + \tilde{\vect \gamma} \cdot \vomega_\avg 
	= (\alpha + \gamma_1 \omega_{\pk,1}) + \vect \beta\cdot \vE + \gamma_2 \omega_{\avg,2} 
	\geq \alpha + \vect \beta \cdot \vE + \vect \gamma \cdot \vomega_\avg $ since we have assumed that $\vomega_\avg \preceq \vomega_\pk$.
	
	For the second statement, assume $\vomega_\pk = \vomega_\avg$ and an optimal solution with $\vect \gamma \leq 0$. Let's define a new solution $\tilde\alpha = \alpha + \vect \gamma \cdot \vomega_\pk$ and $\tilde{\vect \gamma} = \vect 0$. This leaves the objective function \eqref{eq:opti_dual_obj} unchanged, while also satisfying the constraints \eqref{eq:opti_dual_cons} because for all $(\vE,\vomega)\in \setQ_{\vomega_\pk}$, $\tilde \alpha + \vect \beta \cdot \vE + \tilde{\vect \gamma}\cdot \vect \omega = \alpha + \vect \beta \cdot \vE + \vect \gamma \cdot \vomega_\pk \leq H(\vE) $, where the last inequality follows, as above, from the fact that if $(\vE,\vomega)\in \setQ_{\vomega_\pk}$, then $(\vE,\vomega_\pk)\in \setQ_{\vomega_\pk}$. The fact that we can set $\vect \gamma = \vect 0$ when there is no max-average assumption is the dual version of the fact that one can  remove the constraint (\ref{eq:opti_max_av_cons}) in the primal problem (\ref{eq:opti_entropy}). 

\end{proof}

\paragraph{Approximation scheme.} In the dual formulation of the optimisation problem, the constraint \eqref{eq:opti_dual_cons} is difficult to evaluate, even if we have an efficient representation of the set $\mathcal{Q_{\vomega_{\pk}}}$, because the function $H(\vE)$ is non-linear in $\vE$. The central idea behind our algorithm is the observation that the entropy function $H(\vE)$ is concave and that it can therefore be lower-bounded by the pointwise minimum of a finite family of linear functions. 

Specifically, the entropy function is of the form $H(\vE) = \sum_x p(x)h_{bin}(E_x)$, where we used the binary entropy function $h_{bin}(E) = - \sum_{a = \pm 1}\frac{1+aE}{2} \log \frac{1 + aE}{2}$, which is depicted in Figure~\ref{fig:entr_approx}. By dividing the interval $[-1,1]$ in $k$ segments of equal length, computing the value of $h_{bin}(E)$ at the ends of the segments and connecting the dots as in Figure~\ref{fig:entr_approx}, one can find parameters $(c_i,d_i)$ for $1\leq i \leq k$ such that 
\begin{IEEEeqnarray}{rL}
	\label{eq:entr_approx}
	h_{bin}(E)\geq \min_{i} \{ c_i E + d_i \}.
\end{IEEEeqnarray}
\begin{figure}
	\centering
	\includegraphics[scale=1]{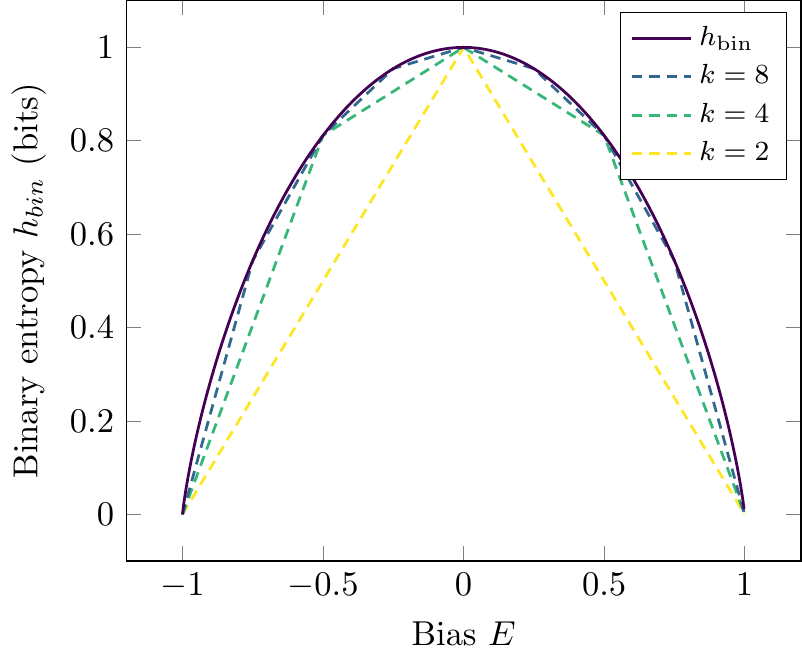}
	\caption{Comparison between the binary entropy and piece-wise linear approximations, for a number of segments equal to $k=2,4,8$.  \label{fig:entr_approx} }
\end{figure}%
This then yields the following lower-bound on $H(\vE)$ 
\begin{IEEEeqnarray}{rL}
	\label{eq:lb_entr}
	H(\vE) = \sum_{x=1}^2 p(x) h_{bin}(E_x) &\geq \min_{(i_1,i_2)}\left\lbrace \sum_{x=1}^2 p(x) (c_{i_x} E_x + d_{i_x})\right\rbrace\\
	 	&= \min_{(i_1,i_2)}\{r_{(i_1,i_2)}+\vect r_{(i_1,i_2)}\cdot \vect E\}\, ,\\
	 	&\equiv H_k(\vE)\label{eq:piece}\,,
\end{IEEEeqnarray}
where $(i_1,i_2) \in \{1,\cdots, k\}^2$, and where we have defined $r_{(i_1,i_2)}=\sum_{x=1}^2 p(x) b_{i_x}$ and $\vect r_{(i_1,i_2)}=(p(1) a_{i_1},p(2) a_{i_2})$. The piecewise linear approximations $H_k(\vE)$  uniformly converge to $H(\vE)$ in the limit $k\rightarrow \infty$.

We can then replace the constraint \eqref{eq:opti_dual_cons} in the dual problem, with the stronger set of constraints
\begin{IEEEeqnarray}{rL}
	\label{eq:opti_dual_cons_new}
	\alpha + \vect \beta \cdot \vE + \vect \gamma \cdot \vomega \leq r_{(i_1,i_2)} + \vect r_{(i_1,i_2)} \cdot \vE\,, \quad \text{ for all }&(\vE,\vomega) \in \setQ_{\vomega_\pk} \IEEEnonumber\\
	\text{ and } &(i_1,i_2) \in \{1, \cdots, k\}^2\,.
\end{IEEEeqnarray}
This set of constraints is stronger because it implies \eqref{eq:opti_dual_cons}. 
 Since they become equivalent to \eqref{eq:opti_dual_cons} in the limit $k\rightarrow \infty$, the value of $H^\star$ can be found up to arbitrary precision. Note also that a similar method could be used in scenarios with more that two outputs $|X|\geq 2$, but the approximation scheme would be more involved.

\paragraph{Semidefinite constraints.} Finally, we need a way to enforce the new set of constraints \eqref{eq:opti_dual_cons_new}, which are of the form 
\begin{IEEEeqnarray}{rL}
	\label{eq:gen_dual_cons}
	\alpha' + \vect \beta' \cdot \vE + \vect \gamma \cdot \vomega \leq 0 \fa (\vE,\vomega) \in \setQ_{\vomega_\pk}
\end{IEEEeqnarray}
with $\alpha' = \alpha - r_{(i_1,i_2)}$ and $\vect \beta' = \vect \beta - \vect r_{(i_1,i_2)}$. Such constraints can be recast in a semidefinite form, using the semidefinite representation of the quantum set $\setQ$ derived in Proposition~\ref{prop:sdp_char}, as shown explicitly in the following proposition.
\begin{prop}
	\label{prop:dual_sdp}
	Let $\alpha' \in \R$ and $\vect{\beta}',\vect{\gamma} \in \R^2$ be given, with $\vect{\gamma}\leq 0$ without loss of generality, following the remark below the dual formulation	. Then the constraint \eqref{eq:gen_dual_cons} is equivalent to the existence of $\vect{\gamma}'\in \mathbb{R}^2$ and  $\vect{\delta} \in \R^4$ such that
	\begin{IEEEeqnarray}{c?c?c}
		\label{eq:sdp_ma}
		A(\alpha',\vect\beta') + C(\vect\gamma+\vect\gamma')+E(\vect{\delta}) \preceq 0,
		& \sum_{i=1}^4 \delta_i + \vect{\gamma}'\cdot \vect{\omega}_{\pk}\geq 0,
		& \vect \gamma '\preceq 0\,,
	\end{IEEEeqnarray}
	where $A(\alpha',\vect\beta')$, $C(\vect\gamma)$, and $E(\vect{\delta})$ are $\R^{4\times 4}$ matrices depending linearly on their arguments as follows	
	\begin{IEEEeqnarray}{c?c}
		A(\alpha',\vect \beta')=\frac{1}{4}\left(
			\begin{IEEEeqnarraybox}[][c]{c,c,c,c}
				\alpha' & 0 & 2\beta'_1 &0\\
				0 & \alpha' & 2\beta'_2 &0\\
				2\beta'_1 & 2\beta'_2 &\alpha' & 0\\
				0 & 0 & 0 &\alpha' 
			\end{IEEEeqnarraybox}
		\right) &
		C(\vect\gamma)=\frac{1}{8}\left(
			\begin{IEEEeqnarraybox}[][c]{c,c,c,c}
				\gamma_1+\gamma_2& 0 & 0 &2\gamma_1\\
				0 & \gamma_1+\gamma_2 & 0 &2\gamma_2\\
				0 & 0 &\gamma_1+\gamma_2 & 0\\
				2\gamma_1 & 2\gamma_2 & 0 &\gamma_1+\gamma_2
			\end{IEEEeqnarraybox}
		\right)\,,
	\end{IEEEeqnarray}		
	and with $E(\vect{\delta})$ the matrix that has $(E)_{ii} =\delta_i$ as unique non-zero entries.
\end{prop}

\begin{proof}
	By Theorem~\ref{prop:sdp_char}, if $(\vE,\vect \omega)\in \setQ_{\vomega_\pk}$, then there exists a $\vomega'$ with the following properties: $i)$ $\vomega'\preceq \vomega$, $ii)$ $\Gamma(\vE,\vect \omega')\geq 0$ where $\Gamma(\vE,\vect \omega)$ is a matrix of the form \eqref{eq:sdp_char}, and $iii)$ $(\vE,\vect \omega')\in \setQ_{\vomega_\pk}$. Using $i)$, $iii)$, and the fact that $\vect{\gamma}\preceq 0$, we have that $\alpha' + \vect \beta' \cdot \vE + \vect \gamma \cdot \vomega \leq  \alpha' + \vect \beta' \cdot \vE + \vect \gamma \cdot \vomega'\leq 0$. Thus checking that the linear constraint $\vect \beta' \cdot \vE + \vect \gamma \cdot \vomega'\leq 0$ holds on the set $\{(\vE,\vect \omega) \text{ s.t. } \Gamma(\vE,\vect \omega)\geq 0 \text{ and } \vect{\omega}\leq \vect{\omega}_\text{pk}\}$ is a sufficient condition for (\ref{eq:gen_dual_cons}). It is also necessary because this set belongs to $\setQ_{\vomega_\pk}$.
	
Expressing $\alpha' + \vect \beta' \cdot \vect E + \vect \gamma \cdot \vect \omega$ as $\Tr[(A(\alpha',\vect\beta')+C(\vect\gamma)) \Gamma(\vE,\vect \omega)]$, the constraint \eqref{eq:gen_dual_cons} is thus equivalent to showing that $\max_\Gamma\Tr[(A(\alpha,\vect\beta)+C(\vect\gamma)) \Gamma]\leq 0$ subject to the constraints $\Gamma \succeq 0$,  $\Tr[E(\delta_i) \Gamma] = 1$ for $1\leq i \leq 4$ and $\Tr[C(\delta_x) \Gamma] \leq \omega_{\pk,x}$ for $1 \leq x \leq 2$. Taking the dual formulation of this SDP, we find that this holds if and only if there exists $\vect \delta \in \R^4$ and $\vect \gamma' \in \R^2$ such that $\vect \gamma' \preceq  0$ and $\sum_i \delta_i +\vect \gamma \cdot \vect \omega_{\mathrm pk} \geq 0$. This establishes \eqref{eq:sdp_ma}.
\end{proof}

\paragraph{Algorithm.} Putting everything together, we have an algorithm that computes a lower-bound $H_k^\star \leq H^{\star}$ on the worst case Shannon entropy, using semidefinite programming. Remember one is given $\vE$, $\vomega_\pk$, $\vomega_\avg$ and $p(x)$. First, fix $k\in \N$ and determine the $k^2$ coefficients $r_{(i_1,i_2)} ,\vect r_{(i_1,i_2)}$ satisfying \eqref{eq:piece} with the method described above. Then use semidefinite programming to find the optimal value of $\alpha + \vect \beta \cdot \vE + \vect \gamma \cdot \vomega_{\avg}$ (depending on the  variables $(\alpha,\vect \beta,\vect \gamma)$), while imposing $k^2$ constraints of the form \eqref{eq:sdp_ma} with $\alpha' = \alpha - r_{(i_1,i_2)}$ and $\vect \beta' = \vect \beta - \vect r_{(i_1,i_2)}$. By Proposition~\ref{prop:dual_sdp}, these constraints are equivalent to \eqref{eq:opti_dual_cons_new}, which is a stronger constraint than the initial dual constraints \eqref{eq:opti_dual_cons} $\alpha + \vect \beta \cdot \vE + \vect \gamma \cdot \vomega$ for all $(\vE,\vomega)\in \setQ_{\omega_{\text{pk}}}$. This implies that $H_k^\star \leq H^\star$.

\addtocounter{rem}{1}
\paragraph{Remark \arabic{rem}.} Though we focused on the Shannon entropy, the above algorithm can also be straightforwardly adapted to bound the min-entropy
$H_{\min}^\star$, or equivalently, the guessing probability $G^\star$. Actually, in this case the optimal guessing probability can be solved using a single SDP. Indeed 
\begin{equation}\label{eq:pw_guess}
G(\vE)=\sum_{x=1}^2 p(x) \max_a \frac{1+aE_x}{2}=\max_{a_1,a_2}\sum_{x=1}^2 p(x) \frac{1+a_xE_x}{2}
\end{equation}
is the exact pointwise maximum
of the four linear functions $\sum_{x=1}^2 p(x) \frac{1+a_xE_x}{2}$ indexed by the four values $(a_1,a_2)\in\{-1,1\}^2$. Thus the dual constraint $\alpha+\vect{\beta}\cdot\vE+\vect{\gamma}\cdot\vect{\omega}\geq G(\vE)$, the analogue of the dual constraint (\ref{eq:opti_dual_cons}), can be exactly expressed as a SDP constraint without involving an approximation scheme (the resulting SDP is then similar to the one introduced in \cite{ref:bss2014,ref:nps2014} in the context of standard Bell scenarios).

\addtocounter{rem}{1}
\paragraph{Remark \arabic{rem}.} The sequence of SDPs for bounding the Shannon entropy has nice convergence properties. Clearly, it convergences to the optimal value: $\lim_{k \rightarrow \infty} H_k^\star =  H^\star$. Furthermore, one gets a strictly increasing sequence when using powers $k = 2^l$, because $H_{2^l}(\vE) \leq H_{2^{l+1}}(\vE)$. However, even for a finite $k\geq 2$, the value $H_k^\star$ has several nice properties. First, the lower-bound $H_k^\star \leq H^\star$, is sufficient to certify the presence of a finite amount of randomness given by $H_k^\star$. Second, whenever the correlations $\vE$ are non-classical, i.e.\ $H^\star >0$, we also have $H_k^\star>0$. This is because the function $H_k(\vE)$ , defined in \eqref{eq:piece}, also has the property that $H_k(\vE) = 0$ if and only if $E_x = \pm 1$, for all $x$ with $p(x)>0$ (see the Remark 1 in Subsection 3.2). Third, the first non-trivial case, corresponding to $k = 2$, already gives a better lower-bound $H^\star\geq H_2^\star$ on $H^\star$, than the one $H^\star\geq H_{\min}^\star=-\log_2 G^\star$ that can be obtained by solving the SDP corresponding to the min-entropy. Indeed, one can see that $H_{\min}^\star\leq H_2^\star$ as follows. First since any feasible solution of the SDP corresponding to $H_{\min}^\star$ is also a feasible solution of the SDP corresponding to $H_2^\star$, one only needs to show that the objective function of the first SDP is always smaller than the objective function of the second SDP. The min-entropy SDP has objective function $H_{\min}=-\log G=-\log_2 \sum_{x,\lambda}p(x)p(\lambda)G(E_x^\lambda)$,  where $G(E_x^\lambda)=\max_a\frac{1+aE_x^\lambda}{2}$. Using  the concavity of the log function, this can be upper-bounded as $-\log G=-\log_2 \sum_{x,\lambda}p(x)p(\lambda)G(E_x^\lambda)\leq -\sum_{x,\lambda}p(x)p(\lambda)\log_2 G(E_x^\lambda)$. Observe that on the interval $E_x\in[-1,1]$, $-\log_2\max_a\frac{1+aE_x}{2}\leq 1-|E_x|$, thus one can further upper-bound the objective function as $-\log G\leq \sum_{x,\lambda}p(x)p(\lambda)(1-|E_x^\lambda|)$. But this last expression is simply the objective function of $H_2^\star$.

\addtocounter{rem}{1}
\paragraph{Remark \arabic{rem}.} Finally, note that instead of fixing the two values $\vE=(E_1,E_2)$, one can also merely fix a linear function $E=c_1 E_1+c_2 E_2$ of them in the above SDPs, and similarly, one can fix a linear function of the two averages energies $\vomega_\avg=(\omega_{\avg,1},\omega_{\avg,2})$. We use this feature in the numerical examples below.

\subsection[Computation of the entropy for several concrete examples]{Computation of the entropy for several concrete\\ examples}
\label{sec:entropy_examples}
We now  illustrate our algorithm by computing the conditional entropy $H^\star$ on several examples. 

Fist, we apply our method to a case where we have a max-average constraint of the form $\vomega_\avg = (\omega,\omega)$, with $\omega  = 0.3$ but no max-peak constraint. We compute the entropy as a function of the violation $E_{-} = \frac{1}{2}(E_1 - E_2)$ of the classical bound \eqref{eq:classical_bound}. When $|E_-|\leq 2 \omega = 0.6$, the behaviour admits a deterministic decomposition so that $H^\star = 0$, but this bound can be violated by quantum devices, because the maximum quantum value is $ 2\sqrt{\omega(1-\omega)} \approx 0.92$ \cite{ref:correnergy}. In Figure~\ref{fig:entropy_violation}, we compute the lower-bounds $H_k^\star \leq H^\star$, for different number of segments $k$ used in the approximation of the binary entropy.

Secondly, we illustrate our algorithm in the more general case where one uses all the measurement statistics to compute the entropy. In Figure~\ref{fig:entropy_domain} we compute the entropy as a function of the correlations $\vE$ for the two different types of assumptions. We take symmetric constraints of the form $\omega_1 = \omega_2 = 0.15$, 

\begin{figure}[t]
	\centering
	\includegraphics[scale=1]{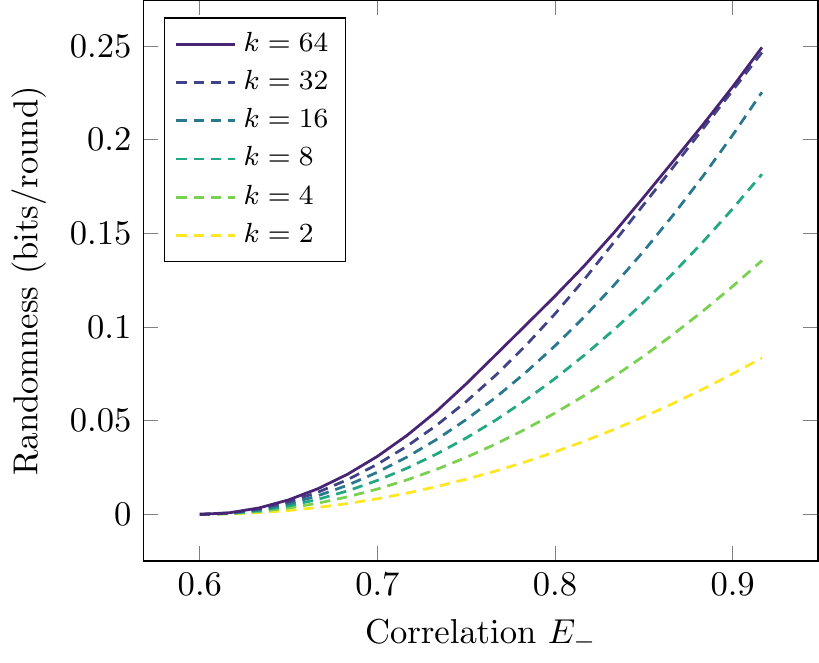}
	\caption{\label{fig:entropy_violation} A converging series of lower-bounds $H_k^\star$ on the worst case conditional entropy $H^\star$, as a function of $E_- = \frac{1}{2}(E_1 - E_2)$, assuming the max-average constraints $\vomega_{\avg} = (\omega,\omega)$ with $\omega = 0.3$ and a uniform input distribution. The maximum quantum value of $E_-$ is $\max_{\vE\in \setQ_{\vomega}} E_- = 2\sqrt{\omega(1-\omega)} \approx 0.92$ but the correlations admit a deterministic decomposition if and only if $|E_-| \leq 0.6$.}
\end{figure}
\begin{figure}[h!]
	\centering
	\begin{subfigure}{0.4\linewidth}
		\includegraphics[scale=1]{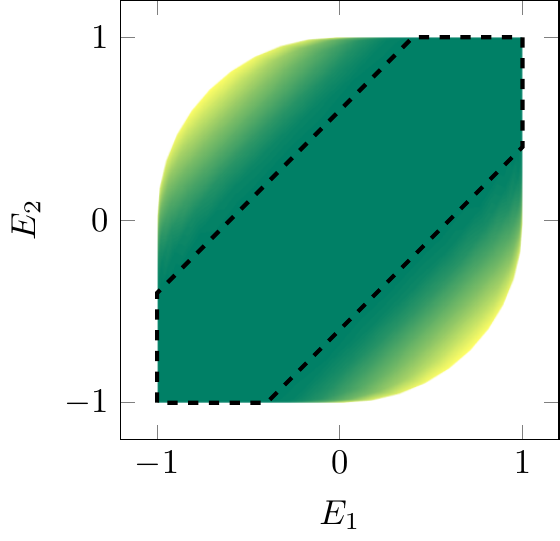}
		\subcaption{Max-average assumption}
	\end{subfigure}%
	\begin{subfigure}{0.5\linewidth}
		\includegraphics[scale=1]{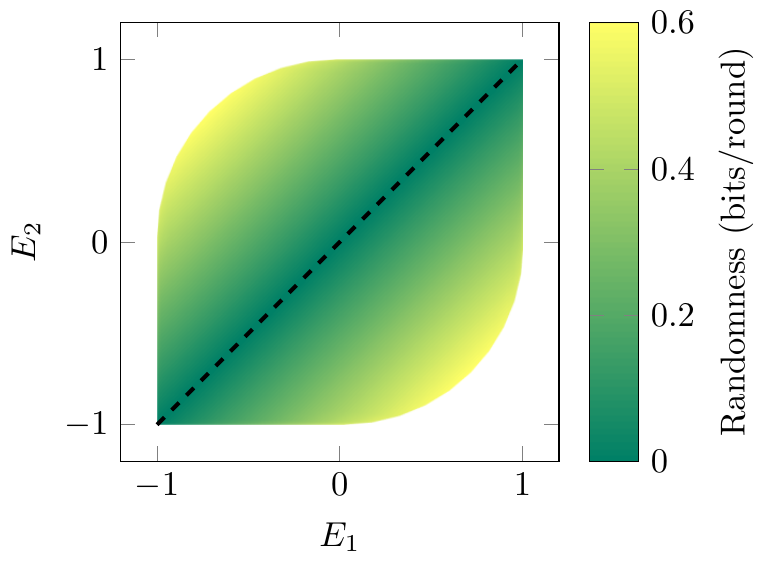}
		\subcaption{Max-peak assumption}
	\end{subfigure}
	\caption{\label{fig:entropy_domain} Entropy $H^\star$ as function the correlation $\vE$, for two different types of assumptions : (a) the max-average assumption $\vomega_\avg = (0.15,0.15)$ (putting trivially $\vomega_\pk = (1,1)$) and (b) the max-peak assumption $\vomega_\pk = (0.15,0.15)$ (putting trivially $\vomega_\avg = (0.15,0.15)$) . The figures were obtained by computing a lower bound $H^\star_k$ on the entropy with an approximation scheme with $k=16$ segments. The dotted regions are the classical regions where $H^\star = 0$ (respectively $|E_1-E_2|\leq 0.6$ and $E_1 - E_2 = 0$), outside of them one can certify a positive rate of randomness generation, given by $H_k^\star >0$.
}
\end{figure}
Let us now apply the algorithm to two experimental implementation proposed in \cite{ref:correnergy}, to show its practical relevance:  the Binary Phase Shift Keying (BPSK) implementation and the On-Off Keying (OOK) implementation.

\paragraph{Binary Phase Shift Keying} The BPSK implementation of \cite{ref:correnergy} is based on displaced coherent states $\ket{\psi_{x}} = \ket{\pm \xi}$, which are defined in the phase space $(X,P)$ of a single mode (with the convention that $[X,P] = i$), and on a binned homodyne measurement of the $X$ quadrature $M = \sgn (X)$. See Figure~\ref{fig:BPSK_scheme} for an experimental implementation using quantum optics components. Taking into account a finite detection efficiency $\eta$, the implementation produces the expected correlations $\vE = \left(\erf\left(\eta\xi\right),-\erf\left(\eta\xi\right)\right)$, while the states have an energy $\bracket{N} = \xi^2/2$. In Figure~\ref{fig:BPSK_entropy}, we compute a lower bound on $H^\star_k \leq H^\star$, using all the measurement statistics $\vE$ and assuming only a bound on the average energy 
\begin{IEEEeqnarray}{rL}
	\label{eq:average_max_average}
	\sum_x \sum_\lambda p(x)p(\lambda)\omega^\lambda_x = \sum_x p(x) \omega_{\avg,x} = \bar \omega_{\avg}\,,
\end{IEEEeqnarray}
where we averaged over the hidden variables $\lambda$ (max-average assumption), as well as over the inputs $x$. In addition to taking into account the noise, we also study the effect of using a safety margin $\delta \geq 0$ on the energy thresholds $\omega$, so we chose $\bar{\omega}_\avg = (1+\delta)\xi^2/2$. The entropies in Figure~\ref{fig:BPSK_scheme} were computed with an approximation scheme with $k=32$ segments and a uniform input distribution. See \cite{ref:correnergy} for a further discussion of the validity of our assumptions for this implementation and the role the the local oscillator.

\paragraph{On-Off Keying} As a last example, we study the On-Off Keying (OOK) implementation of \cite{ref:correnergy}. The correlations are obtained by sending the coherent states $\ket{\psi_1} = \ket{0}$ or $\ket{\psi_2} = \ket{\xi}$ and using a single photon detector with efficiency $\eta$. The output is labelled $1$ if the detector clicks and $-1$ otherwise. The expected correlations are $E_1 = -1$ and $E_2 = 1 - 2e^{-\eta\xi^2/2}$ and the energies are $\vomega_\pk = (0,\xi^2/2)$. It turns out that for this implementation one needs the stronger max-peak assumption, since the max-average assumption alone gives a zero rate (the correlation $\vE$ is in the classical set, see \cite{ref:correnergy}). The upside of this implementation is that, when applied to it, our analysis tolerates arbitrary small detection inefficiencies. This implementation admits a direct analytical formula for the entropy $H^\star$, since the observed correlations $\vE$ are on the border of the set $\setQ_{\vomega}$ with $\vomega = (0,\omega_2)$ and so there is a unique way to decompose $\vE$ into extremal points of $\setQ_{\vomega}$ (these are $\vE^1 = (-1,-1)$ and $\vE^2 = (-1, -1 + 2\omega_2)$). We find
\begin{IEEEeqnarray}{rL}
	H^\star = p_X(1)\frac{1+E_2} {2\omega_{\pk,2}} h_{bin} \left(\omega_{\pk,2}\right)\,.
\end{IEEEeqnarray}
The entropy is shown in Figure \ref{fig:entropy_OOK} for different regimes of operation $\omega_{\pk,1} = \xi^2/2$ and different detector efficiencies $\eta$.

\clearpage
\begin{figure}
	\centering
	\includegraphics[scale=1]{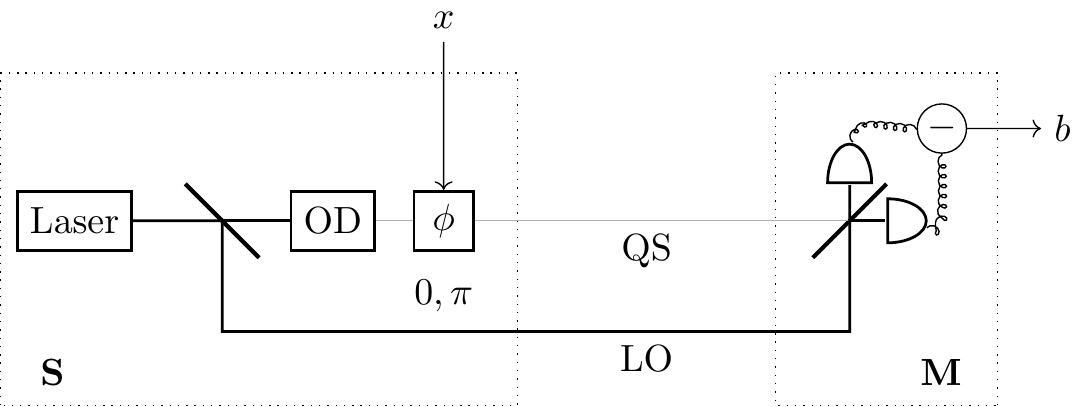}
	\caption{Binary Phase Shift Keying (BPSK) implementation. A highly attenuated laser beam (OD is an optical density) is send though a phase shifter ($\phi$) which is controlled by the input $x$ and applies the phase $0$ or $\pi$. This produces the Quantum Signal (QS) (one of the two coherent states $\ket{\psi_x} = \ket{\pm \xi}$) which is send to the measurement device. A homodyne measurement of the $X$ quadrature is then performed by interference with a Local Oscillator (LO) which was previously extracted from the laser. The final output is $b = \sgn (X)$ (Figure taken from \cite{ref:correnergy}.)
	\label{fig:BPSK_scheme}}
\end{figure}
\begin{figure}[h]
	\centering
	\includegraphics[scale=1]{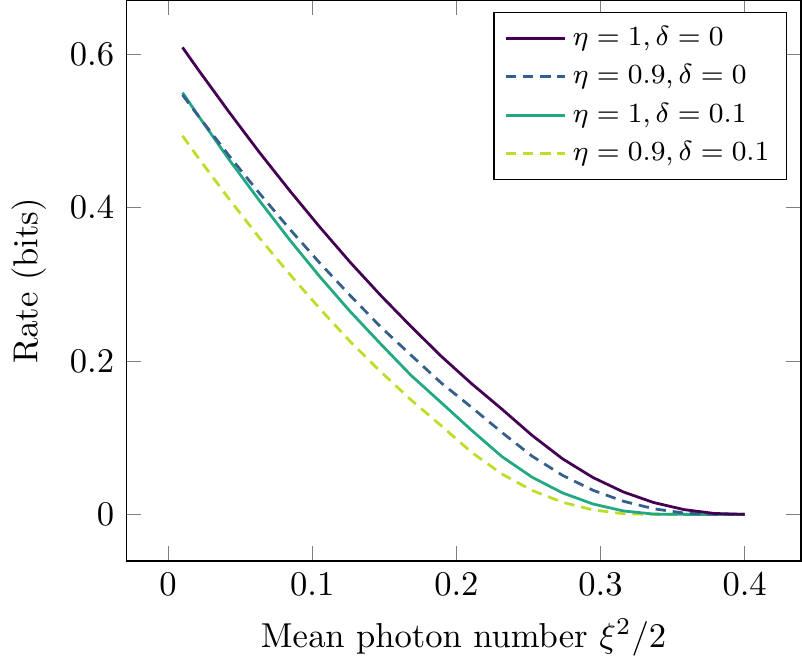}
	\caption{Lower-bound $H_k^\star$ on the conditional Shannon entropy $H(A|X\Lambda)$ for the BPSK implementation as a function of the mean photon number $\xi^2/2$ of the implementation. We analyse the rate for different detector efficiencies $\eta$ and security margins $\epsilon$ on the energy bound. The rates correspond to lower-bounds $H_k^\star$ computed with an approximation scheme with $k=32$ segments, using the assumption \eqref{eq:average_max_average} and a uniform input distribution. The energy threshold was chosen conservatively as $\omega = (1+\delta)\xi^2/2$. Note that the rate is larger in the low energy regime. This is because the coherent states are close to the vacuum so the output is almost unbiased.
	\label{fig:BPSK_entropy}}
\end{figure}
\clearpage
\begin{figure}[t]
	\centering
	\includegraphics[scale=1]{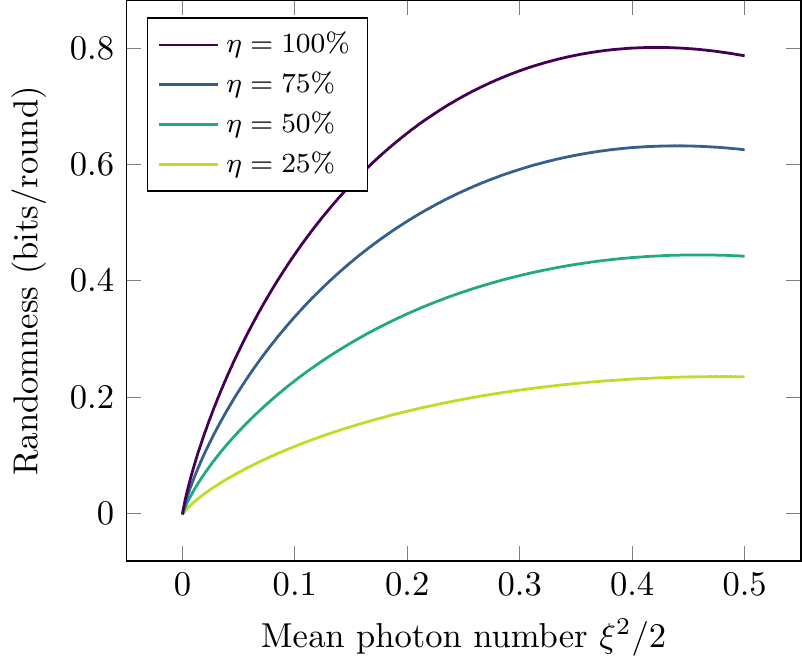}
	\caption{ \label{fig:entropy_OOK} Worst case conditional Shannon entropy $H^\star$ for the OOK implementation using a max-peak assumption for different regimes of operation and detector efficiencies.}
\end{figure}


\section{Protocol for randomness certification}
\label{sec:security_proof}

In this section, we analyse the randomness in the case where a device is used multiple rounds consecutively and we construct an explicit protocol for testing the device and certifying the output randomness.

\subsection{General model and assumptions}
\label{subsec:assumptions}
We consider the prepare-and-measure device defined in the previous sections, when it is used $n$ times successively. The classical random variables observed by the user of the device are the inputs $X^n = (X_1\cdots X_n)$ and the outputs $A^n = (A_1 \cdots A_n)$%
\footnote{%
	As is standard practice, from now on we denote random variable by uppercase letters and the values they take by lowercase letters.
}%
. The classical information that a potential adversary, Eve, has on the device is represented by a random variable $\Lambda$.  The correlation between the inputs, outputs, and Eve's information is represented by a probability distribution $\mu_{A^n X^n \Lambda}$ unknown to the user of the protocol.

We make the following assumptions.
\begin{itemize}
	\item \emph{Choice of input distribution}: The input $X_i$ at round $i$ is generated independently of the past $W_i=(A^{i-1},X^{i-1})$ and of $\Lambda$ and with identical distribution $p_X$ for all $i$:
\begin{equation}\label{eq:ass1}
\mu_{X_i|W_i\Lambda}(x)=p_X(x)\,.
\end{equation}	

	\item \emph{Existence of a valid quantum representation for each round $i$ conditioned on the past}: The output $A_i$ at round $i$ originates from a device used with input $X_i$ and characterized by a valid quantum behaviour $(\vect E_i(W_i,\Lambda),\linebreak[1]\vect \omega_i(W_i,\Lambda))\in \mathcal{Q}$. We can thus in particular write
	\begin{equation}\label{eq:ass2}
	\mu_{A_i|X_i,W_i\Lambda}	(a|x) = \frac{1}{2}(1+aE_{x,i}(W_i,\Lambda))\,.
	\end{equation}
	As the notation indicates, the correlations $\vect E_i(W_i,\Lambda)$ and the energies $\vect \omega_i(W_i,\Lambda)$ can depend on the past $W_i$ and on $\Lambda$.
	
	\item \emph{Max-peak assumption}: We assume a max-peak bound on the energies
	\begin{IEEEeqnarray}{rL}
		\label{eq:maxpk_protocol}
		\vomega_i(W_i,\Lambda) \leq \vomega_{\pk} \fa i = 1,\cdots ,n\,.
	\end{IEEEeqnarray}
	This assumption is simply the direct transcription of \eqref{eq:maxpeak}, the idea that there is an absolute energy limit satisfied by the states prepared by the source. More generally, we could also require \eqref{eq:maxpk_protocol} to hold only with high probability or on a large fraction of the rounds, but this would make our subsequent analysis more cumbersome.
	
	\item \emph{Max-average assumption}: 	Finally, we assume that for some energy thresholds $\vect \omega_\avg$ and for some (small) $\epsilon_{\vect \omega} \geq 0$, 
	\begin{IEEEeqnarray}{rL}
		\label{eq:maxavg_protocol}
		\Pr\left(\frac{1}{n}\sum_{i=1}^n \vect{\omega}_i(W_i,\Lambda)\preceq \vect{\omega}_\avg\right)\geq 1-\epsilon_{\vect \omega} \,.
	\end{IEEEeqnarray}
	This is the non-i.i.d. version of the max-average assumption \eqref{eq:maxav}, expressing the fact that the energies can fluctuate from one round to the other, provided that the overall average over $n$ rounds stays bounded. Note that we require the bound to hold only with high probability $1-\epsilon_{\vect \omega}$, because we want our analysis to cover simple i.i.d.\ strategies for Eve where she chooses at each run with probability $p(\lambda)$ energies $\vect{\omega}_\lambda$ satisfying $\sum_\lambda p(\lambda) \vect{\omega}_\lambda\leq \vomega_\avg-\delta$, for some security margin $\vect \delta \succeq 0$. If Eve follows such a strategy it is expected  that $\frac{1}{n}\sum_{i=1}^n \vect{\omega}_i\preceq \vect{\omega}_\avg$ only with high probability, so it may happen that $\frac{1}{n}\sum_{i=1}^n \omega_{i,x}\succeq \omega_{\avg,x}$ for some $x$, albeit with very small probability if $n$ is large.
\end{itemize}

The first two assumptions are entirely similar to their counterparts in device-independent (DI) protocols with classical-side information \cite{ref:kzb2017,ref:zkb2018}. The max-peak assumption simply constrains the set of quantum behaviours of the devices at the level of individual runs and is thus not fundamentally different than the no-communication assumption in DI QRNG.

The novelty of our randomness estimation analysis, on the other hand, lies in the max-average assumption which constrains the mean behaviour of the devices over $n$ rounds, and not at the individual level, where it can arbitrary fluctuate. This assumption cannot be directly used in existing randomness estimation frameworks, so we provide a new one by generalizing the techniques from \cite{ref:kzb2017,ref:zkb2018} to take into account fluctuating energy. 

Note that, although the thresholds $\vomega_\avg$ can be chosen based on some partial knowledge of the source coming from a theoretical model, it also conceivable to estimate them before the experiment by measuring the average energy using a trusted energy meter. For optical applications, like the ones in Section~\ref{sec:entropy_examples}, this amounts to placing a photo-detector between the source and the measurement apparatus and measuring the average number of photons of the pulses (and assuming some stability overt time of the average energy output of the source).

In the following, we present our security analysis in the general case where one uses both a max-peak and max-average assumption, but note that it also applies when using only one of the two assumption. Without loss of generality, we also assume the bounds of eq. (\ref{eq:ass_energies}).

\subsection{Trade-off Functions and randomness estimation} 

The main tool we use to estimate randomness are Trade-off Functions (TF).

\begin{definition}
	\label{def:TF} 
	Let $p(x)$ be given. We say that $(\vect \alpha,\vect \beta,\vect \gamma) \in (\R^2)^3$ is a \emph{Trade-off Function} with max-peak energies $\vomega_\pk$ if $\vect \gamma \leq 0$ and 
	\begin{IEEEeqnarray}{l}
		\label{eq:TF}
		\alpha + \vect \beta \cdot \vE + \vect \gamma \cdot \vect \omega \leq H(\vE) \fa (\vect E,\vect \omega)\in \setQ_{\vomega_\pk} \,,
	\end{IEEEeqnarray}
	where $\alpha = \sum_x \alpha_x$.
\end{definition}
A Trade-off Function is a linear lower bound on $H(\vE)$ that holds for any quantum behaviour that satisfies the max-peak assumption $\setQ_{\vomega_\pk} = \{(\vE,\vomega) \in \setQ| \vomega \preceq \vomega_\pk \}$. It is therefore a feasible point of the dual optimisation problem \eqref{eq:opti_dual}, used in Section \ref{sec:entropy} to compute a lower-bound on $H^\star$. Such a TF can be found and optimized for a specific use with the algorithm described in Section~\ref{sec:entropy}. Note that there is a small additional degree of freedom since we define $\alpha = \alpha_1+\alpha_2$; this will be used below.

Such TF are closely related to other functions used in Device-Independent security proofs to characterize the randomness as a function of the correlations. The min-tradeoff functions in \cite{ref:arv2016,ref:adfrv2018} are linear lower-bounds on the conditional von Neumann entropy. The Probability Estimation Factors in \cite{ref:kzb2017,ref:zkb2018} are stronger that the TF but reduce tot them in the limit $\beta \rightarrow 0$. Finally in \cite{ref:nbsp2016}, the randomness bounding functions are a convex lower-bound on the surprisal $-\log p(a|x)$ for a subset of the inputs.

In a randomness generating protocol the left hand side of \eqref{eq:TF} has to be determined to get the rate of randomness generation. To evaluate the value of $\alpha + \vect \beta \cdot \vE$, we define the unbiased estimator 
\begin{IEEEeqnarray}{rL}
	\label{eq:estimator}
	\xi(a,x) = \frac{1}{p(x)}(\alpha_x + a \beta_x)\, ,
\end{IEEEeqnarray}
which satisfies $\esp{\xi(A,X)} = \alpha + \vect \beta \cdot \vE$, when $\esp{A|x} = E_x$. The value of $\alpha + \vect \beta \cdot \vE$ is then estimated by computing 
\begin{equation}
\bracket{\xi}_{A^nX^n} = \frac{1}{n}\sum_i \xi(A_i,X_i)\,.
\end{equation}

The following theorem is the central result of this section. We derive a lower-bound on the surprisal $-\log \mu(A^n|X^n\Lambda)$ of the outcome $A^n$ given $X^n\Lambda$ as a function of the value of estimator $\bracket{\xi}_{A^nX^n}$ and the energy upper-bounds  $\vomega_\pk$, $\vect \omega_{\avg}$. These are variables in the hands of the user. 

\begin{theorem}
	\label{prop:rand_TF} 
	Let $\epsilon_t >0$, let the distribution $\mu_{A^nX^n\Lambda}$ satisfy the assumptions of section \ref{subsec:assumptions} and let $(\vect \alpha,\vect \beta,\vect \gamma)$ be a Trade-off Function with max-peak energies $\vomega_\pk$, then the bound
	\begin{IEEEeqnarray}{rL}
		\label{eq:rand_TF}
		-\frac{1}{n}\log \mu(A^n|X^n,\Lambda) \geq \bracket{ \xi }_{A^nX^n} + \vect \gamma \cdot\vect\omega_\avg - t \,,
	\end{IEEEeqnarray}
	holds with a probability greater than $1- \epsilon_\omega - \epsilon_t$. The error term is
	\begin{IEEEeqnarray}{rL}
		\label{eq:error_term}
		t = \sqrt{2V}\sqrt{\frac{\log(1/\epsilon_t)}{n}} + \frac{\xi^+}{3} \frac{\log(1/\epsilon_t)}{n}
	\end{IEEEeqnarray}
with $V = \max\{(\xi^+ + \bar\gamma)^2,(\xi^- + \bar\gamma)^2\} + 2 \max\{\log |A|(\bar \gamma + \xi^-),0)  + \frac{4|A|}{e^2}(\log_2 e)^2$, where $\xi^+ = \max_{a,x} \xi(a,x)$, $\xi^- = \min_{a,x} \xi(a,x)$, $\bar{\gamma}= \sum_x \gamma_x$, and where $|A|$ is the cardinality of the random variable $A$. 
\end{theorem}

Roughly speaking, a lowerbound on the surprisal $-\log \mu(A^n|X^n\Lambda) \geq k$ certifies the presence of $k$ bits of randomness in the outputs and is directly proportional to the length $|K|$ of a uniform key that can be extracted from the raw output string $A^n$. So this theorem establishes the relation between the amount of randomness and the observed behaviour, when the device is used $n$ times. This will be made more precise in Subsection~\ref{subsec:protocol_sec_proof}. 

Defining the rate as the key length per round $R = |K|/n$, Theorem~\ref{prop:rand_TF} then directly provides the leading terms of the rate $R$ as a function of $n$: first a leading constant term $\bracket{ \xi }_{A^n,X^n} + \vect \gamma \cdot\vect\omega_\avg$, which is the value of the Trade-off Function evaluated at the observed behaviour and which gives the asymptotic rate, and then a sub-leading error term given by $- \sqrt{2V}\sqrt{\frac{\log(1/\epsilon_t)}{n}}$, which scales as $O(1/\sqrt n)$. 

Theorem~\ref{prop:rand_TF} follows from the following new construction, which was inspired by the Probability Estimation Factors in \cite{ref:kzb2017}. Let $(\vect E,\vect \omega)\in Q$ be some behaviour and let the random variables $A X$ have the distribution $p(a,x) = \frac{p(x)}{2}(1+aE_x)$. Then we can define the random variable
\begin{IEEEeqnarray}{rL}
	\label{eq:def_T}
	T = \xi(A,X) + \vect \gamma \cdot \vect \omega +\log p(A|X)\, ,
\end{IEEEeqnarray}
Defined as such, the variable $T$ satisfies the following two lemmae.

\begin{lemma}
	\label{prop:T_bound}
	Let $(\vE,\vect \omega) \in \setQ_{\omega_\pk}$ and let $(\vect \alpha,\vect \beta,\vect \gamma)$ be a TF with max-peak energies $\omega_\pk$. Then the variable $T$ defined in \eqref{eq:def_T} satisfies 
	\begin{IEEEeqnarray}{rL}
		\esp{T}\leq 0 \,.
	\end{IEEEeqnarray}
\end{lemma}
\begin{proof}
	This follows directly from the definition of a TF \eqref{eq:TF}. Using the equalities $\esp{\xi(A,X)} = \alpha + \vect \beta \cdot \vE$ and $\esp{-\log p(A|X)} = -\sum_{a,x}p(x)p(a|x)\log p(a|x)=H(\vE)$, we find that
	\begin{IEEEeqnarray}{rL}
		\esp{T} &= \alpha + \vect \beta \cdot \vE + \vect \gamma \cdot \vect \omega - H(\vE) \leq 0 \,.
	\end{IEEEeqnarray}		

\end{proof}

\begin{lemma}
	\label{prop:V_bound}
	Under the same assumptions as Lemma~\ref{prop:T_bound}, the variable $T$ defined in \eqref{eq:def_T} satisfies 
	\begin{IEEEeqnarray}{r?c?l}
		\esp{T^2}\leq V & \text{ and } & T \leq \xi^+ \,,
	\end{IEEEeqnarray}
	with $V$ and $\xi^+$ defined as in the statement of Theorem~\ref{prop:T_bound}.
\end{lemma}

\begin{proof}
	The bound $T \leq \xi^+$ holds because of $\xi(A,X) + \vect \gamma \cdot \vect \omega +\log p(A|X) \leq \xi(A,X) \leq \xi^+$, where we used that fact that $\vomega \succeq 0$ and $\vect \gamma \preceq 0$. 
	
	To prove the second bound, we first show that $\esp{T^2}\leq \esp{T'^2}$, where $T' = \xi(A,X) + \bar\gamma +\log p(A|X)$. To see this, we observe that $\Delta = T-T' = \vect \gamma \cdot \vect \omega - \bar\gamma\geq 0$ (because we have assumed that $\vomega\preceq \vomega_\pk \preceq \vect 1$). This indeed entails that $\esp{T'^2} = \esp{(T-\Delta)^2} = \esp{T^2} + \Delta^2 - 2\Delta\esp{T} \geq \esp{T^2}$, where we used $\esp{T}\leq 0$ from Lemma~\ref{prop:T_bound}. 
	
	Next, we expand the square as a sum of three terms $\esp{T'^2} = \esp{ (\xi(A,X)+\bar \gamma)^2} \allowbreak$ $+ 2 \esp{(\xi(A,X)+\bar \gamma)\log p(A|X)} + \esp{(\log p(A|X))^2}$ and bound each term individually :
	\begin{IEEEeqnarray}{rL}
		\esp{ (\xi(A,X)+\bar \gamma)^2} 
					& = \sum_{ax} p(a,x) (\xi(a,x)+\bar\gamma)^2				\\
					& \leq \max\{(\xi^+ +\bar \gamma)^2,(\xi^- + \bar \gamma)^2\} \,, 	\\
		\esp{(\xi(A,X)+ \bar \gamma)\log p(A|X)} 
					& = \sum_{ax} p(x) p(a|x)\log p(a|x) (\xi(a,x)+ \bar \gamma)\\
					& \leq -H(A|X) (\xi^- + \bar \gamma)\\
					& \leq \max \{ -\log |A| (\xi^- + \bar \gamma),0 \} \\
		\esp{(\log p(A|X))^2}
					& \leq \sum_{ax} p(x) p(a|x)(\log p(a|x))^2\\
					& \leq \frac{4|A|}{e^2}(\log_2 e)^2\,.
	\end{IEEEeqnarray}
	We have used the inequalities $ p(a|x)\log p(a|x) \leq 0$ , $H(A|X) \leq \log_2|A|$, as well as $p(a|x)(\log_2 p(a|x))^2 \leq \frac{4 }{e^2}(\log_2 e)^2$. This concludes the proof of the lemma.
\end{proof}

To complete the proof of Theorem~\ref{prop:rand_TF}, we need the following Hoeffding type concentration inequality for super-martingales, in addition to the two lemmae.

\begin{prop}[Equation (18) in \citep{ref:fgl2012}] \label{prop:ci1}
	Let $\left(T_i\right)$ be a sequence of random variables, with $i \in \{0,\cdots n\}$ that (a) satisfies the property of a supermartingale difference : $\esp{T_i|T_1^{i-1}}\leq 0$, for all $i$, and (b) is such that $T_i\leq \xi^+$ and $\esp{(T_i)^2|T_1^{i-1}} \leq V$, for all $i$, then 
	\begin{IEEEeqnarray}{rL}
		\Pr \left( \frac{1}{n} \sum_{i=1}^n T_i \geq t \right) \leq \epsilon_t
	\end{IEEEeqnarray}
	with $t = \sqrt{2V}\sqrt{\frac{\log(1/\epsilon_t)}{n}} + \frac{\xi^+}{3} \frac{\log(1/\epsilon_t)}{n}$.
\end{prop}

With this concentration inequality, we can finally prove Theorem~\ref{prop:rand_TF}.

\begin{proof}[Proof of Theorem~\ref{prop:rand_TF}]
	Let us define the random variables
	\begin{IEEEeqnarray}{rL}
		T_i = \xi(A_i,X_i) + \vect \gamma \cdot \vect \omega _i(W_i,\Lambda)) + \log \mu(A_i|X_i;W_i,\Lambda)\, ,
	\end{IEEEeqnarray}		
	for $i\in \{0, \cdots n\}$.
By the first two assumptions on the devices, eqs. (\ref{eq:ass1}) and (\ref{eq:ass2}), we can write
$E[T_i|W_i,\Lambda]=\vect\alpha+\vect\beta\cdot \vE_i(W_i,\Lambda)+\gamma \cdot \vect \omega _i(W_i,\Lambda))+H(\vE_i(W_i,\Lambda))$. Using the Lemmas \ref{prop:T_bound} and \ref{prop:V_bound}, we then find that $\esp{T_i|W_i,\Lambda} \leq 0$, $\esp{T_i^2|W_i,\Lambda}\leq V$, and $T_i \leq \xi^+$, for all $i$. Since the variables $T_1^{i-1}$ are a function of $W_i$ and $\Lambda$, we can apply the Hoeffding type bound in Proposition~\ref{prop:ci1}, which states that $\Pr \left( \frac{1}{n} \sum_{i=1}^n T_i \leq t \right) \geq 1-\epsilon_t$. Next, we rewrite the sum $\sum_i T_i$, using the definition of $T_i$ and using the following two relations: (1) $\bracket{\xi}_{A^nX^n} = \frac{1}{n}\sum_i \xi(A_i,X_i)$ and (2) $ \sum_i \log \mu(A_i|X_i;W_i,\Lambda) = \log \prod_i \mu(A_i|X_i;W_i,\Lambda) = \log (A^n|X^n,\Lambda)$. We find that
	\begin{IEEEeqnarray}{rL}
		 \Pr \left(
		 	- \tfrac{1}{n}\log \mu(A^n|X^n,\Lambda) 
		 	\geq \bracket{\xi}_{A^n,X^n} 
		 		+ \tfrac{1}{n}\sum_i \vect \gamma \cdot \vect \omega_i(W_i,\Lambda)
		 		- t \right)\geq 1-\epsilon_t \,.
	\end{IEEEeqnarray}	
At last, we combine this with the upper-bound on the average energy \eqref{eq:maxavg_protocol}, which states that $\Pr\left(\frac{1}{n}\sum_{i=1}^n \vomega_i(W_i,\Lambda)\preceq \vect{\omega}_\avg\right)\geq 1-\epsilon_{\vect \omega}$ and the fact $\vect \gamma \preceq 0$, to replace (probabilistically) $\tfrac{1}{n}\sum_i \vect \gamma \cdot \vect \omega_i$ by $\gamma\cdot \vect{\omega}_\avg$ . Using the bound $\Pr(A\cap B) \geq \Pr(A) + \Pr(B) -1$, we find that the inequality \eqref{eq:rand_TF} holds with a probability greater than $1- \epsilon_\omega - \epsilon_t$ as required.
\end{proof}

\subsection{Protocol and security proof}
\label{subsec:protocol_sec_proof}
We now use Theorem \ref{prop:rand_TF} to prove that the following protocol is sound.
\begin{algorithm}[H]
\caption{A protocol for randomness certification based on an energy constraint}
\label{box:protocol}
\begin{flushleft}
\begin{algorithmic}[1]
	\STATEx \textbf{Arguments}
		\STATEx \hspace{\algorithmicindent}
		- Number of measurement rounds  $n$.
		
		\STATEx \hspace{\algorithmicindent}
		- Binary input distribution $p(x)>0$.
		
		\STATEx \hspace{\algorithmicindent}
		- Energy thresholds $\vect \omega_{\pk}\in \R^2$ \eqref{eq:maxpk_protocol} and $\vomega_{\avg} \in \R^2$ \eqref{eq:maxavg_protocol} with $\epsilon_\omega>0$.
		
		\STATEx \hspace{\algorithmicindent}
		- Security parameters $\epsilon_t, \epsilon_m, \epsilon_{Ext} >0 $ and $\epsilon = \epsilon_t + \epsilon_m + \epsilon_{Ext} + \epsilon_\omega$.
		
		\STATEx \hspace{\algorithmicindent}
		- Trade-off function $\vect \alpha, \vect \beta, \vect \gamma \in \R^2$ (Definition~\ref{def:TF}), with $\xi(a,x)$ \eqref{eq:estimator} and $t,V,\xi^+$~\eqref{eq:error_term}. 
		
		\STATEx \hspace{\algorithmicindent}
		- Threshold $r$, such that $r - t \leq 1$.
		
		\STATEx \hspace{\algorithmicindent}
		- Strong extractor $\mathcal{E}$ with parameters $(n,l,\sigma_h,\sigma,\epsilon_{Ext})$ where the bound on the min-entropy is
		\begin{IEEEeqnarray}{rL}
			\sigma_h = n\left( r  - \sqrt{2V}\sqrt{\frac{\log(1/\epsilon_t)}{n}} - \frac{\xi^+}{3} \frac{\log(1/\epsilon_t)}{n} - \frac{\log(1/\epsilon_m)}{n} \right) \, .
		\end{IEEEeqnarray}	
	\STATEx
		
	\STATE Repeat steps [2-3] $n$ times, with $i \in \{1,\cdots n\}$;
	\STATE \hspace{\algorithmicindent} Generate input $X_i$;
	\STATE \hspace{\algorithmicindent} Use device and record output $A_i$;
	\STATE Determine whether $\bracket{\xi}_{A^n X^n} + \vect\gamma\cdot\vect \omega_\avg \geq r$;
	\STATE \hspace{\algorithmicindent} if not, abort;
	\STATE \hspace{\algorithmicindent} if yes, apply extractor $\mathcal{E}$ with uniform seed $S$: $K = \mathcal{E}(A^n,S)$;
\end{algorithmic}
\end{flushleft}
\end{algorithm}

Protocol~\ref{box:protocol} is a standard randomness generation protocol that first involves a test to verify if a chosen estimator $\bracket{\xi}_{A^n X^n} = \frac{1}{n} \sum_i \xi(A_i,X_i) \geq r$ is greater than some predetermined threshold $r$. Typically, the choice of estimator and of the threshold $r$ will be made by solving the optimization problem \eqref{eq:opti_dual_obj} using an expected behaviour $\vE$ for the device. This is based on some prior information, such the way it was designed or an estimation obtained by sampling the behaviour a finite number of times. If the test is passed, then by Theorem~\ref{prop:rand_TF}, we have a bound on the probability of occurrence of $A^n$ of the form $-\log \mu(A^n|X^n\Lambda)\gtrsim nr$, which holds with high probability, independently of how the device actually behaves and up to the error term $t$. Formally, this is expressed as a lower-bound on the smooth min-entropy of $A^n$, which is made precise in Proposition~\ref{lm:esupe_fail_esmaxprob} below.

Conditioned on the passing of the test, we apply a strong extractor $\mathcal{E}$ to the raw output string $A^n$, using a uniform seed $S$, to produce the key $K$. A strong extractor depends on five parameters $(n,l,\sigma_h,\sigma,\epsilon_{Ext})$, where $n$ is the length (in bits) of the input random string, $l$ is the length of the additional (and typically small) seed, $\sigma_h$ is a lower-bound on the min-entropy of the input random string, $\sigma$ is the length of the output string, and $\epsilon_{Ext}$ denotes how close the final string is to uniform (in trace  distance), see \cite{vadhan_pseudorandomness_2012}. There exist various constructions for strong extractors, which in the best case can extract about $\sigma\approx \sigma_h$ random bits, up to some corrections.

We show below, mostly following \cite{ref:kzb2017,ref:zkb2018}, that the resulting protocol is $\epsilon$-sound.

\begin{theorem}
	\label{prop:soundness}
	Let the distribution $\mu$ of $A^nX^n\Lambda$ satisfy the assumptions of section \ref{subsec:assumptions}, and assume valid arguments for Protocol~\ref{box:protocol}. Let $\Pass$ stand for the event $\bracket{\xi}_{A^n,X^n} + \vect \gamma \cdot \vomega_\avg \geq r$ and denote its probability $\kappa = \Pr(\Pass)$. Then the final string $K = \mathcal{E}(A^n,S)$ of Protocol~\ref{box:protocol} is $\epsilon/\kappa$-close in trace distance to a uniform random string independent of the seed ($S$), the inputs ($X^n$), and the adversary's information on the device ($\Lambda$):
	\begin{IEEEeqnarray}{rL}
		\label{eq:soundness}
		\TV \left( \mu_{[KSX^n\Lambda|\mathrm{Pass}]},\mathrm{Unif}_{K \mathcal{S}}\otimes\mu_{[X^n\Lambda|\mathrm{Pass}]}\right)
			\leq \epsilon_{Ext} + (\epsilon_\omega + \epsilon_t + \epsilon_m)/\kappa\leq \epsilon/\kappa\,.
	\end{IEEEeqnarray}
	In particular $\Pr(\mathrm{Pass})\times \TV \left( \mu_{[KSX^n\Lambda|\mathrm{Pass}]},\mathrm{Unif}_{K \mathcal{S}}\otimes\mu_{[X^n\Lambda|\mathrm{Pass}]}\right)\leq \epsilon$, i.e., this defines a $\epsilon$-sound randomness generation protocol for which the probability of both passing the test and deviating from an ideal distribution is guaranteed to be small, see \cite{ref:pm2014}.
\end{theorem}
In the above statement, the trace or total variation distance is defined as $\TV (\mu,\nu) = \frac{1}{2}\sum_x |\mu(x)-\nu(x)|$, for two distributions $\mu$ and $\nu$ of a random variable $X$. 

The proof of Theorem~\ref{prop:soundness} is done in two steps. First, in Proposition~\ref{lm:esupe_fail_esmaxprob}, we characterise the randomness in the raw string $A^n$, by bounding on the smoothed $X^n\Lambda$-conditional min-entropy of $A^n$, after conditioning on $\Pass$. In Definition~\ref{def:smoothminentropy} we give precise definitions for two variants of the smoothed conditional min-entropy, which are related by Lemma~\ref{prop:markov_argument}. Finally, we use this to prove Theorem~\ref{prop:soundness}.

\begin{definition}
	\label{def:smoothminentropy}
	Let $\mu$ be a distribution of $AZ$. 
	
	The smooth \emph{average} conditional min-entropy $H^{\epsilon}_{\min,\mu}(A|Z)$ is the maximum $k$ for which there exists a distribution $\nu$ of $AZ$, ($i$) with the same marginals $\mu[Z] = \nu[Z]$, ($ii$) such that $\TV(\nu,\mu)\leq\epsilon$, and ($iii$) with $-\log_2 \esp{\max_{a}\nu(a|z)} \geq k$ for all $a,z$. 
	
	The smooth \emph{worst-case} conditional min-entropy $H^{u,\epsilon}_{\min,\mu}(A|Z)$ is the maximum $k$ for which there exists a distribution $\nu$ of $AZ$, (1) with the same marginals $\mu[Z] = \nu[Z]$, (2) such that $\TV(\nu,\mu)\leq\epsilon$ and (3) with $-\log_2 \nu(a|z) \geq k$ for all $a,z$. 
\end{definition}

Note that, in the present definition of the smooth average conditional min-entropy, the requirement of equal marginals is not standard practice, but it leads to a slightly better security parameters for the protocol, when one chooses to impose equal marginals in the definition of soundness. The \emph{average} and \emph{worst-case} variants of the smooth conditional min-entropy in Definition~\ref{def:smoothminentropy} are related by the following standard lemma.

\begin{lemma}[Lemma 5 \cite{ref:kzb2017}]
	\label{prop:markov_argument}
	Let $\mu$ be a distribution of $AZ$ with $H_{\min,\mu}^{\epsilon_1}(A|Z) \geq \sigma$ and let $\epsilon_2>0$, then $H_{\min,\mu}^{u,\epsilon_1 + \epsilon_2} \geq \sigma - \log(\frac{1}{\epsilon_2})$. 
\end{lemma}

\begin{prop}
	\label{lm:esupe_fail_esmaxprob}
 	Under the same assumptions as in Theorem~\ref{prop:soundness}, the distribution $\mu_{\Pass} = \mu_{[A^nX^n\Lambda|\Pass]}$, conditioned on the passing of the test, admits the following bounds on the smoothed $X^n\Lambda$-conditional min-entropies of $A^n$.
 	\begin{IEEEeqnarray}{rL}
 		H^{(\epsilon_\omega+\epsilon_t)/\kappa}_{\min,\mu_{\Pass}}(A^n|X^n\Lambda) &\geq n(r -t)-\log_2\tfrac{1}{\kappa} 
 		\label{eq:entropy_bound_av}\\
 		H^{u,(\epsilon_\omega+\epsilon_t+\epsilon_m)/\kappa}_{\min,\mu_{\Pass}}(A^n|X^n\Lambda) &\geq n(r -t)-\log_2\tfrac{1}{\epsilon_m} = \sigma_h
 		\label{eq:entropy_bound_wc}\,,
	\end{IEEEeqnarray}
	where the error term $t$ is defined as in \eqref{eq:error_term}.
\end{prop}

\begin{proof}
	To simplify the notation, we write in this proof $A=A^n$ and $Z=(X^n,\Lambda)$. We first show the bound on the \emph{average} smooth conditional min-entropy \eqref{eq:entropy_bound_av}. We construct a distribution $\nu$ of $AZ$ that witnesses the claim $\esp{ {\max}_a \nu(a|x)} \leq 2^{-n(r-t)}/\kappa$ as follows. We define $\nu(az) = \mu(z|\Pass)\nu_z(a)$, so that it has the same $Z$ marginals as $\mu_{\Pass}$. The distribution $\nu_z(A)$ is obtained from a subnormalized distribution $\tilde \nu_z(A)$ defined as 
	\begin{IEEEeqnarray}{rL}
		\tilde \nu_z(a) = \mu(a|z,\Pass) \knuth{\phi(az)},
	\end{IEEEeqnarray}
	where $\knuth{\rho}$ for a logical expression $\rho$ denotes the $\{0,1\}$-valued function evaluating to 1 iff $\rho$ is true, and where $\phi(az)$ is the logical expression $-\frac{1}{n}\log \mu(a|z)\geq \bracket{\xi}_{a,z} + \vect\gamma\cdot\vect\omega_\avg-t$ or equivalently $\mu(a|z)\leq 2^{-n(\bracket{\xi}_{a,z} + \vect\gamma\cdot\vect\omega_\avg-t)}$. In other words, $\tilde \nu_z(a)$ is the conditional distribution $\mu(a|z,\Pass)$ but with the bad events removed. From Theorem~\ref{prop:rand_TF}, we have that $\mu(\phi) = \sum_{a,z} \mu(a,z) \phi(az) \geq 1 - \epsilon_t-\epsilon_\omega$. We also denote $P(az)$ the logical expression $\bracket{\xi}_{a,z} + \vect \gamma \cdot \vomega_\avg \geq r$, that establishes if the test is passed or not for a given input and output string. It satisfies $\mu(P) = \mu(\Pass) = \kappa$. We can then derive the following upper-bound
	\begin{IEEEeqnarray}{rL}
		\tilde \nu_z(a) &= \mu(a|z,\Pass) \knuth{\phi(az)}\\
			&= \frac{\mu(a,z)\knuth{P(a,z)}}{\mu(\Pass,z)} \knuth{\phi(az)}\\
			&= \frac{\mu(a|z)}{\mu(\Pass |z)}\knuth{\phi(az)}\knuth{P(a,z)}\\			
			&\leq  2^{-n(r-t)} \frac{1}{\kappa_z}
	\end{IEEEeqnarray}
where $\kappa_z = \mu(\Pass|z)$. Note that $2^{-n(r-t)}/\kappa_z\geq 2^{-n(r-t)}\geq 2^{-n}$, we can thus apply Lemma~2 in \cite{ref:kzb2017} to obtain the distributions $\nu_z(A)$ such that $\nu_z(A)\geq \tilde\nu_z(A)$, $\nu_z(A)\leq 2^{-n(r-t)}/\kappa_{z}$, and $\TV(\nu_z,\mu_{[A|z,\Pass]})\leq 1-w_{z}$, where $w_z = W(\tilde \nu_z(A))$ is the weight of the subnormalized distribution $\tilde \nu_z(A)$. 

As stated above, we can now define $\nu(az)=\nu_{z}(a)\mu(z|\Pass)$. Using an expression for the TV distance of distributions with same maginals (Equation~2 in \cite{ref:kzb2017}), we find
\begin{IEEEeqnarray*}{rL}
	\TV \left( \nu,\mu_{\Pass} \right)
    	&= \sum_{z}\TV \left( \nu_{z}, \mu_{[A|z,\Pass]}\right) \mu(z|\Pass)\\
    	&\leq \sum_{z}(1-w_{z})\mu(z|\Pass)\\
    	&= 1-\sum_{z} W(\tilde\nu_z(A))\mu(z|\Pass)\\
    	&= 1-\sum_{z}\sum_{a}\frac{\mu(az)}{\mu(z|\Pass)\mu(\Pass)}\knuth{\phi(az)}\knuth{P(az)}\mu(z|\Pass)\\
    	&= 1-\sum_{z}\sum_{a}\mu(az)\knuth{\phi(az)}\knuth{P(az)}/\mu(\Pass)\\
    	&= 1- \mu(\phi|\Pass)\\
    	&\leq \mu(\bar \phi)/\mu(\Pass)\\
    	&\leq (\epsilon_\omega+\epsilon_t)/\kappa
    	\label{eq:lm:esupe_fail_esmaxprob_1}
\end{IEEEeqnarray*}
For the average maximum probability of $\nu$, we get
\begin{IEEEeqnarray}{rL}
    \esp{{\max}_a \nu(a|Z)}
    	&=\sum_{z}\mu(z|\Pass)\max_{a}\nu_{z}(a)\notag\\
    	&\leq 2^{-n(r-t)}\sum_{z}\mu(z|\Pass)/\kappa_{z}\notag\\
    	&= 2^{-n(r-t)}\sum_{z}\mu(z)/\kappa = 2^{-n(r-t)}/\kappa,
\end{IEEEeqnarray}
which establishes that $H^{(\epsilon_\omega+\epsilon_\delta)/\kappa}_{\min,\mu_{\Pass}}(A|Z) \geq n(r-t)-\log_2\frac{1}{\kappa}$.

	We now treat the \emph{worst-case} smooth conditional min-entropy. Using Lemma~\ref{prop:markov_argument} with $\epsilon_1 = (\epsilon_t + \epsilon_\omega)/\kappa$ and $\epsilon_2 = \epsilon_m/\kappa$, we deduce that
	\begin{IEEEeqnarray}{rL}
		H_{min,\mu_{\Pass}}^{u,(\epsilon_\omega + \epsilon_t + \epsilon_m)/\kappa}(A|Z) \geq n(r-t)-\log_2\tfrac{1}{\epsilon_m} = \sigma_h\, ,
	\end{IEEEeqnarray}
	where we removed the dependence of the amount of entropy on $\kappa$. 
\end{proof}

The Proposition~\eqref{lm:esupe_fail_esmaxprob} allows us to complete the security proof of Theorem~\ref{prop:soundness}.

\begin{proof}[Proof of Theorem~\ref{prop:soundness}]
	Using the bound on the worst-case smooth conditional min-entropy \eqref{eq:entropy_bound_wc}, we deduce that there must exist a distribution $\nu$  of $A^nX^n\Lambda$ that satisfies $\nu_{[X^n\Lambda]} = \mu_{[X^n\Lambda|\Pass]}$, with $\TV(\nu_{[A^nX^n\Lambda]},\mu_{[A^nX^n\Lambda|\Pass]})\leq(\epsilon_\omega + \epsilon_t + \epsilon_m)/\kappa$ and $-\log\max_{a^n}\nu(a^n|x^n \lambda) \geq \sigma_h$, for all $x^n,\lambda$. Using a strong extractor $\mathcal{E}:\{0,1\}^n \times \{0,1\}^l \rightarrow \{0,1\}^\sigma$, which satisfies the extractor constraints with entropy $\sigma_h$ and security parameter $\epsilon_{Ext}$, we find that, for all $x^n,\lambda$, $\TV(\nu[KS|x^n,\lambda],\mathrm{Unif}_{KS}) \leq \epsilon_{Ext}$, where $K = \mathcal{E}(A^n,S)$ is the final key and $S$ a uniform seed. Since we have equal marginals $\nu_{[X^n\Lambda]} = \mu_{[X^n\Lambda|\Pass]}$, we can extend this to
	\begin{IEEEeqnarray}{rL}
		\TV(\nu_{[KSX^n\Lambda]},\mathrm{Unif}_{KS}\otimes \mu_{[X^n\Lambda|\Pass]}) \leq \epsilon_{Ext}.
	\end{IEEEeqnarray}
	On the other hand, by the data processing inequality, $\TV(\mu_{[KSX^n\Lambda|\Pass]},\nu_{[KSX^n\Lambda]}) \leq \TV(\nu_{[A^nX^n\Lambda]},\mu_{[A^nX^n\Lambda|\Pass]}) \leq (\epsilon_\omega + \epsilon_t + \epsilon_m)/\kappa$. We conclude by the triangular inequality and find that
	\begin{IEEEeqnarray}{rL}
		\TV(\mu_{[K S X^n\Lambda|\Pass]},\mathrm{Unif}_{K \mathcal{S}}\otimes\mu_{[X^n\Lambda|\Pass]})
			\leq \epsilon_{Ext} + (\epsilon_\omega + \epsilon_t + \epsilon_m)/\kappa\,.
	\end{IEEEeqnarray}
\end{proof}

\section{Conclusion}
\label{sec:conlusion}

In this paper, we have performed a complete analysis of a QRNG protocol based on the semi-device-independent scheme that was introduced in \cite{ref:correnergy}. Our results have been used in the experimental implementation of such a QRNG that was recently reported in \cite{ref:maxavQRNG}. With respect to previous semi-device-independent QRNG proposals, we have presented an efficient finite-statistic analysis that takes into account arbitrary shared randomness, statistical fluctuations of the devices, and memory effects, and which is based on a natural, physical hypothesis -- the energy constraints.

Our analysis implicitly assumes that the device has no quantum memory, a reasonable and realistic assumption in the semi-device-independent setting and given the status of current technology. This assumption appears in two places in our analysis.

First, in the fact that the source and the measurement device are not allowed to share prior entanglement. This is used in Section~\ref{sec:sdp} to characterize the set $\mathcal{Q}$ of quantum behaviours, and thus also in the computation of the entropy bounds in Section~\ref{sec:entropy}. Preliminary numerical explorations that we have carried out show that allowing entanglement would enlarge the quantum set, resulting in slightly lower entropy bounds. Shared entanglement could actually arise quite naturally in set-ups where the source sends a local oscillator to the measurement device, such as in Figure~\ref{fig:BPSK_scheme}, and where this mode is slightly entangled with the signal states. It would thus be interesting to generalize our results in this direction.

The second place where we implicitly assume that the device has no quantum memory is in the randomness analysis of Section~\ref{sec:security_proof}, where we consider an adversary with classical-side information, i.e., not entangled with the internal quantum systems of the device. We believe that it should be possible to generalize the randomness estimation techniques against quantum-side information introduced in \cite{knill_quantum_2018} to our energy constrained setting.

The above open questions aim at reducing the assumptions used to estimate the randomness produced of the specific scheme introduced in \cite{ref:correnergy} with binary inputs and outputs. Another direction for future research would be to consider more general randomness generating schemes based on energy constraints. In particular, a first natural generalisation would be to increase the number of outcomes. This is especially useful for implementations based on a homodyne measurement, such as the BPSK implementation in Figure~\ref{fig:BPSK_scheme}. In the present analysis the continuous measurement result has to be binned into positive/negative values, this works well, but a finer discretization of the quadrature would yield more randomness. Note that, in such a semi-DI analysis, one would also require more inputs as follows from the results of \cite{ref:ibn2019}. It would be interesting to know if there is a maximal amount of randomness that can be certified under an energy assumption in the limit of an infinite number of inputs and outputs. 

The randomness analysis of a energy-constrained QRNG protocols that we have introduced in Section~\ref{sec:security_proof} is generic and would apply to any scheme for which one can compute Trade-off Functions. The introduction of new protocols with more inputs and outputs would thus merely require a characterization of the corresponding quantum set and a corresponding way to compute Trade-off Functions, i.e., a modification of Sections~\ref{sec:sdp} and \ref{sec:entropy}. Unfortunately, there is no systematic way to do this for semi-DI scenarios, unlike in Bell-scenario, where one can resort to the NPA hierarchy \cite{ref:npa2008}. A possible approach would be to extend the mapping between our semi-device-independent scenario and the standard CHSH scenario presented in Appendix~\ref{sec:ent} to other scenarios. Note that if one has a semidefinite characterisation of the quantum set, then one can readily use the approximation algorithm of Section~\ref{sec:entropy} (extended to more outputs) to compute Trade-off Functions. 

Finally, the generation of certified randomness is one of the most immediate task to consider in a DI or semi-DI setup. It would be interesting to design and prove the security of more complex semi-DI protocols based on energy constraints, such as quantum key distribution. 	

\section*{Acknowledgements}

We thank Yanbao Zhang for interesting discussions.
We acknowledge support from the EU Quantum Flagship project QRANGE.
T.V.H. is supported by a FRIA grant from the Fond National de la Recherche Scientifique (Belgium). SP is a Senior Research Associate of the Fonds de la Recherche Scientifique - FNRS.

\appendix
\section*{Appendix}
\addcontentsline{toc}{section}{Appendices}
\renewcommand{\thesubsection}{\Alph{subsection}}

\subsection{Mapping to a Bell scenario}\label{sec:ent}

We now provide an explicit mapping between our prepare-and-measure scenario (with an energy assumption) and a Bell scenario (with a no-communication assumption). This provides an alternative explanation for the appearance of the SDP constraint (\ref{eq:sdp_char}) in our context.

Consider a standard Bell scenario, see Figure~\ref{fig:eb_setup}, with two binary measurement per party, characterized by the four correlators $\bracket{ A_x B_y } = \Tr[\rho_{AB} A_x B_y]$ for $x,y\in\{1,2\}$. Here $x$ and $y$ denote the two possible measurements by Alice and Bob, $A_x$ and $B_y$ are the corresponding quantum observables (with $A_x^2=B_y^2=I)$ and $\rho_{AB}$ is a bipartite state shared between Alice and Bob. We denote a tuple $\langle \mathbf{AB}\rangle=\left(\langle A_1B_1\rangle,\langle A_1B_2\rangle,\langle A_2B_1\rangle,\langle A_2B_2\rangle \right)$ specifying a value for each of the four correlators as a Bell behaviour and denote $\mathcal{Q}_{Bell}$ the set of all quantum Bell behaviours. 
\begin{figure}[h]
	\centering
	\includegraphics[scale=1]{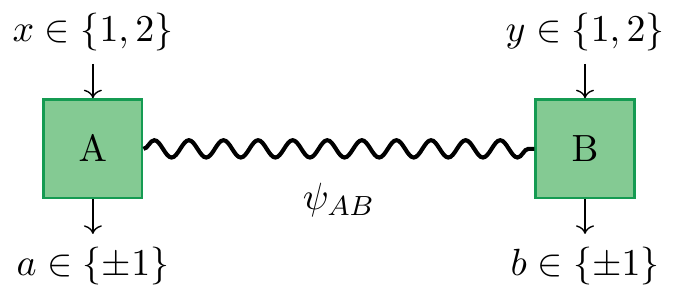}
	\caption{Bell scenario based on a no-communication assumption. We show that the correlations in this scenario are closely related to the ones in the prepare-and-measure scenario based on an energy assumption.\label{fig:eb_setup}}
\end{figure}

\begin{prop}
\label{theorem:isomorphism}
A prepare-and-measure behaviour $(\vect{E},\vect{\omega})$ is in $\setQ$ if and only if there exist a Bell behaviour $\bracket{\mathbf{AB}} \in \setQ_{Bell}$ such that
	\begin{IEEEeqnarray*}{rL}
		\IEEEyesnumber*
		\label{eq:isom}
		\IEEEyessubnumber*
		\langle A_x B_1 \rangle &= E_x\\
		\langle A_x B_2 \rangle &\leq 2\omega_x-1.
	\end{IEEEeqnarray*}
\end{prop}
 The trick to prove Proposition~\ref{theorem:isomorphism} is to view $O$ as a second observable on Bob's side. 

\begin{proof}
To show the equivalence between the sets $\setQ$ and $\setQ_{Bell}$ through the above mapping, it is sufficient to consider extremal behaviours.
Starting from a prepare-and-measure extremal behaviour $(\vE,\vect{\omega})$, we first show that there exists a Bell behaviour satisfying the relations \eqref{eq:isom}. For this, consider a representation for $(\vE,\vect{\omega})$ as in Lemma~\ref{lemma1}. Let $\ket{\phi_+}=\frac{1}{\sqrt{2}}\left(\ket{00}+\ket{11}\right)$ be the maximally entangled two qubit state and let us define the following measurements: $A_x=2\rho_x- 1=\vect{n}_x\cdot\vsigma$, $B_1=M=\vect{m}\cdot\vsigma$ and $B_2=2H-1=\vect{\omega}\cdot\vsigma$. The state $\ket{\phi_+}$ and the measurements $A_x$, $B_y$ define a quantum representation of a Bell behaviour. We find
	\begin{equation}	
		\label{eq:isomorphism}
		\bracket{ A_x B_y }_{\phi_+}= \frac{1}{2} \Tr[A_x B_y]= \Tr[\rho_x B_y]
	\end{equation}
	The first equality follows from the so-called swap trick and the second from the property $\Tr[B_y]=0$. This entails that $\bracket{A_x B_1}=E_x$ and $\bracket{A_x B_2}=2\Tr[\rho_x O]-1 \leq 2 \omega_x -1$, as in (\ref{eq:isom})
	
	For the proof of the converse, we use a celebrated result by Tsirelson \cite{ref:t1987} about extremal correlations in the quantum set $\mathcal{Q}_{Bell}$ : if $\bracket{\mathbf{AB}}$ is an extremal point in the (convex) set $\mathcal{Q}_{Bell}$, then it can be realized with a maximally entangled two-qubit state $\ket{\phi_+}$ and four qubit measurements $A_x,B_y$ with $\Tr[A_x]=\Tr[B_y]=0$ and $A_x^2=B_y^2=\id$. Now consider the prepare-and-measure qubit strategy defined by $\rho_x=(1+A_x)/2$, $M=B_1$, $O=(1+B_2)/2$. We have that $\rho_x$ and $O$ are rank-1 projectors hence they define, respectively, valid pure states and an energy operator. We find that
\begin{equation}
E_x=\Tr[\rho_x M]=\Tr[A_x B_1]=\langle A_xB_1\rangle
\end{equation}
and
\begin{equation}
2\omega_x-1 = \Tr[\rho_x (2H-1)] = \Tr[A_x B_2]=\langle A_xB_1\rangle\,.
\end{equation}
This leads to a valid prepare-and-measure strategy satisfying the relations \eqref{eq:isom}.
\end{proof}

Tsirelson showed that a Bell behaviour $\langle \mathbf{AB}\rangle$ is quantum if and only if there exist two real numbers $u,v$ such that
\begin{IEEEeqnarray}{rL}
	\label{eq:sdp_bell}
		\Gamma_{Bell}=\left(\begin{IEEEeqnarraybox}[][c]{c?c?c?c}
			1 & u & \langle A_1B_1\rangle &  \langle A_1B_2\rangle\\
			  & 1 &	\langle A_2B_1\rangle& \langle A_2B_2\rangle\\
			  &   & 1      & v\\
			  &   &        & 1
		\end{IEEEeqnarraybox}\right) \succeq 0	   \,.
\end{IEEEeqnarray}
Theorem~1 can then also be viewed as a consequence of this SDP characterization and the above mapping.

Interestingly, under this mapping, there is also a direct link between the classical set in our prepare-and-measure scenario and the classical set in the standard Bell scenario. Indeed, the two linear inequalities $|E_1 - E_2|\leq 2(\omega_1 + \omega_2)$ that bound the classical set \cite{ref:correnergy} are equivalent to the two CHSH inequalities $\pm(\bracket{A_1 B_1} -  \bracket{A_2 B_1}) - \bracket{A_1 B_2} - \bracket{A_2 B_2}\leq 2$ in the space of Bell correlators.

\subsection{Properties of the optimisation problem}

\label{sec:apprendix}

Let $\mathcal{S} \subset \R^{\dim(S)}$ be a (non-empty compact) convex set, let $f$ be a  continuous function over $\mathcal{S}$ and consider the following optimisation problem:
\begin{IEEEeqnarray}{R'r'L}
	\label{eq:opti_gen}
	\IEEEyesnumber
	\IEEEyessubnumber*
	{f^\star}(\vect x_0)=&
		 \min_{\{ \vect x^\lambda, p(\lambda) \} }
			& {\sum}_\lambda p(\lambda) f(\vect x^\lambda) \\
		& \text{subject to }
			& {\sum}_\lambda p(\lambda) \vect x^\lambda = \vect x_0\\
			&& \vect x^\lambda \in \mathcal{S}\\
			&& p(\lambda)_\lambda \in \mathcal{P}(\Lambda),
\end{IEEEeqnarray}
where the number of hidden variables $|\Lambda|$ is a priori unbounded. We recover \eqref{eq:opti_entropy} by setting $\vect x_0=(\vE,\vect{\omega}_{\text{avg}})$, $\vect x^\lambda=(\vE^\lambda,\vect{\omega}^\lambda)$, $\mathcal{S}=\mathcal{Q}_{\vect{\omega}_{\text{pk}}}$, $f(\vect x)=f(\vE,\vect{\omega})=f(\vE)=-\sum_{b,x} p(x)\frac{1+bE_x}{2}\log\frac{1+bE_x}{2}$. Similarly, one recovers the variant of \eqref{eq:opti_entropy} corresponding to the guessing probability with $f(\vE)=-\sum_{x} p(x)\max_b\frac{1+bE_x}{2}$, where the minus sign has been introduced to turn the maximization of the guessing probability in the minimization form \eqref{eq:opti_gen}. 

We proceed by showing some general properties of the optimisation problem \eqref{eq:opti_gen}. First note that, while there is no limitation on the number of hidden variables $\lambda$ in the optimisation problem \eqref{eq:opti_gen}, we can show by a simple argument that $\dim(\mathcal{S}) + 1$ are sufficient to reach the minimum. This implies in particular that it was correct to use a minimum instead of a infimum in \eqref{eq:opti_gen}.

\begin{prop}
	\label{prop:finite}
	There exists an optimal solution of \eqref{eq:opti_gen} with $|\Lambda| \leq \dim(\mathcal{S}) + 1$.
\end{prop}

\begin{proof}
	 Let $\mathcal F = \{(\vect x,f(\vect x))|\vect x \in \mathcal{S}\} \subset \R^{\dim(\mathcal{S})+1}$ be the graph of the function $f$ on $\mathcal{S}$ and ${\mathcal F}^\star$ its convex closure. Let $\vect x_0\in \mathcal{S}$ and let $\{ \vect x^\lambda,p(\lambda)|\lambda\in \Lambda \}$ be an optimal solution to the problem \eqref{eq:opti_gen}, with $p^\lambda >0$ for $\lambda \in \Lambda$ and with $|\Lambda|> \dim(\mathcal{S})+1$. Then by construction, the point $\vect v_0 = (\vect{x_0},{f^\star}(\vect x_0)) \in \R^{\dim(\mathcal{S})+1}$ is a convex combination of the points $\vect v^\lambda = (\vect{x}^\lambda,f(\vect x^\lambda)) \in  \mathcal{F}$, so that $\vect v_0 \in {\mathcal F^\star}$. Moreover $\vect v_0$ it is on the border of ${\mathcal F}^\star$ because, by the optimality, for all $\epsilon>0$, $(\vect x_0,f^\star(\vect x_0)- \epsilon)\notin {\mathcal{F}^\star}$. 
	 
	We now use the supporting hyperplane to ${\mathcal{F}^\star}$ at $\vect v_0$, i.e., the fact that there exists an affine function $s:\R^{\dim(\mathcal{S})+1}\rightarrow \R$, such that $s[\vect v]\geq 0$ for all $\vect v \in {\mathcal{F}^\star}$ and $s[\vect v_0] = 0$. By linearity we have $0 = \sum_\lambda p(\lambda) s[\vect v^\lambda]$, but since $s[\vect v^\lambda] \geq 0$ for $\vect v^\lambda \in {\mathcal{F}^\star}$, this implies $s[\vect v^\lambda] = 0$ for all $\lambda \in \Lambda$. Thus all the points $\vect v^\lambda$ live in a subspace of dimension $\dim(\mathcal{S})$. By Carathéodory's theorem, we thus find that $\vect v$ can be expressed as a (possibly different) convex combination of a subset $\Lambda' \subset \Lambda$  of size $|\Lambda'| \leq \dim(\mathcal{S})+1$ of the points $\vect v^\lambda$. 
\end{proof}

This result cannot be used, however, to limit a priori the number of hidden variables $\lambda$, because it does not tell us which finite set of extreme points should be considered for a given $\vect x_0$.
When the function $f$ is concave, however, the following straightforward property can help.

\begin{prop}
	If the function $f$ is concave, there exists an optimal solution of \eqref{eq:opti_gen} with $\vect x_\lambda \in \extr(\mathcal{S})$ extremal for all $\lambda\in \Lambda$.
\end{prop}

Thus if the extremal points of $\mathcal{S}$ are known and finite and the function $f$ is concave, the problem \eqref{eq:opti_gen} reduces to the search of  a finite number of optimal weights $p(\lambda)$, i.e., to a linear program. In the case of the computation of the conditional entropy or of the guessing probability, the function $f(\vect x)$ is concave; however, the number of extreme points of $\mathcal{S}=\mathcal{Q}$ is not finite.

\begin{prop}
	\label{prop:opti_dual_gen}
	For $\vect x_0 \in \mathcal{S}$, the optimisation problem \eqref{eq:opti_gen} admits the dual problem
	\begin{IEEEeqnarray}{R'r'L}
		\label{eq:opti_dual_gen}
		\IEEEyesnumber
		\IEEEyessubnumber*
		{f^\star_{dual}}(\vect x_0) =& 
			\sup_{\{t, \vect t\}}
				& t + \vect t \cdot \vect x_0\\
			&\st 
				& t + \vect t \cdot \vect x \leq f(\vect x) \fa \vect x \in  \mathcal{S} \label{eq:opti_dual_gen_cons}
	\end{IEEEeqnarray}
	which strong duality, i.e., ${f^\star_{dual}}(\vect x_0)={f^\star}(\vect x_0)$. Moreover the supremum becomes a maximum for $x_0 \in \interior(\mathcal S)$.
\end{prop}

\begin{proof}
	First note that any feasible point $t,\vect t$ of \eqref{eq:opti_dual_gen} provides a lower-bound on ${f^\star}(\vect x_0)$, because ${f^\star}(\vect x_0) = {\sum}_\lambda p^\lambda f(\vect x^\lambda) \geq {\sum}_\lambda p^\lambda (t + \vect t \cdot \vect x^\lambda) = t + \vect t \cdot \vect x_0$.
	
	To show strong duality, we first treat the case $\vect x_0 \in \interior(\mathcal S)$ and construct a feasible solution of the dual problem that achieves $f^\star(\vect x_0)$. Consider the epigraph of $f^\star$, defined as $F^+ = \{ (\vect x,y)\in \mathcal S \times \R| y \geq f^\star(x) \}$. This is a convex set, since $f^\star(\vect x)$ is convex in $\vect x$ by construction, so we can consider the supporting hyperplane at $\vect v_0 = (\vect x_0 ,f^\star(\vect x_0)) \in F^+$. We have
	\begin{IEEEeqnarray}{rL}
		\label{eq:firstcon}
		t' + \vect t'\cdot \vect x - s y &\leq 0 \text{\,, for all } (\vect x,y) \in F^+\\
		\label{eq:seccon}
		t' + \vect t'\cdot \vect x - s y &= 0 \text{\,, for }(\vect x,y) = (\vect x_0,f^\star(\vect x_0))\,.
	\end{IEEEeqnarray}
	for some $(t',\vect t',s)$. Because of the definition of $F^+$, we must necessarily have $s \geq 0$ as $F^+$ is unbounded in the direction of increasing $y$, and, since $\vect x_0 \in \interior(\mathcal S)$, we must have $s \neq 0$. We can thus define $t = t'/s$, $\vect t = \vect t'/s$. The first condition \eqref{eq:firstcon} implies that $t,\vect t$ is a feasible solution: $t + \vect t\cdot \vect x \leq f^\star(\vect x) \leq f(\vect x)$ for all $\vect x\in \mathcal S$. The second implies that the supremum is obtained: $t + \vect t\cdot \vect x_0 = f^\star(\vect x_0)$.	
	
	When $x_0 \in S \backslash \interior(S)$ in on the border of $S$, the supporting hyperplane could in principle be vertical such that $s=0$, so that the supremum is not obtained for some finite $(t,\vect t)$. But let's show that we can approach it arbitrarily well. Let $\epsilon>0$, and let $\vect v_0^\epsilon = (\vect x_0,f^\star(x_0)-\epsilon) \notin F^+$. There must exist a separating hyperplane between the convex sets $F^+$ and the point $\vect v_0^\epsilon$, so that 
	\begin{IEEEeqnarray}{rL}
		t' + \vect t'\cdot \vect x  - s y &< 0 \text{\,, for all } (\vect x,y) \in F^+\\
		t' + \vect t'\cdot \vect x - s y &> 0 \text{\,, for }(\vect x,y) = (\vect x_0, f^\star(\vect x_0)-\epsilon)\,,
	\end{IEEEeqnarray}
	Because of the definition of $F^+$, we must have $s \geq 0$ and we cannot have $s=0$. Proceeding as above we find a feasible point $(t,\vect t)$ such that $t + \vect t\cdot \vect x_0 > f(\vect x_0) - \epsilon$. This proves strong duality.
\end{proof}

\end{document}